%% file: main.tex
\documentclass[11pt, onecolumn]{IEEEtran}

\usepackage{times}

\usepackage{amsthm}
\usepackage{amsmath}    
\usepackage{amssymb}    
\usepackage{mathrsfs}   
\usepackage{epsfig}     
\usepackage{graphicx}

\usepackage[numbers,sort&compress]{natbib}

\usepackage{enumerate}

\newcommand{\script}[1]{{\mathscr #1}}

\newcommand{\cC}{{\cal C}}

\newcommand{\cK}{{\cal K}}

\newcommand{\cM}{{\cal M}}
\newcommand{\cN}{{\cal N}}

\newcommand{\cP}{{\cal P}}

\newcommand{\cR}{{\cal R}}

\newcommand{\cT}{{\cal T}}

\newcommand{\cV}{{\cal V}}

\newcommand{\sC}{\script{C}}
\newcommand{\sD}{\script{D}}

\newcommand{\bX}{\mathbf{X}}
\newcommand{\bY}{\mathbf{Y}}

\newcommand{\bx}{\mathbf{x}}
\newcommand{\by}{\mathbf{y}}
\newcommand{\bm}{\mathbf{m}}
\newcommand{\bc}{\mathbf{c}}

\newcommand{\bone}{{\mathbf{1}}}

\renewcommand{\le}{\leqslant}
\renewcommand{\leq}{\leqslant}

\renewcommand{\geq}{\geqslant}

\newcommand{\rank}{\text{rank}}

\newtheorem{thm}{Theorem} 

\newtheorem{lem}{Lemma}
\newtheorem{prop}{Proposition}
\newtheorem{rmk}{Remark}

\newtheorem{definition}[thm]{Definition}

\newcommand{\be}[1]{\begin{equation}\label{#1}}
\newcommand{\ee}{\end{equation}}
\newcommand{\eq}[1]{(\ref{#1})}

\newcommand{\Tref}[1]{Theo\-rem\,\ref{#1}}
\newcommand{\Pref}[1]{Pro\-po\-si\-tion\,\ref{#1}}
\newcommand{\Lref}[1]{Lem\-ma\,\ref{#1}}
\newcommand{\Cref}[1]{Co\-ro\-lla\-ry\,\ref{#1}}

\newcommand{\Eref}[1]{(\ref{#1})}

\newcommand{\aln}[1]{\begin{align*}#1\end{align*}}

\newcommand{\al}[1]{\begin{align}#1\end{align}}

\newcommand{\ep}{\varepsilon}

\newcommand{\F}{\mathbb{F}_p}

\newcommand{\N}{{\mathcal N}}

\newcommand{\R}{{\mathbb R}}

\newcommand{\ciid}{\ol C_{\iid}}

\newcommand{\iid}{\text{i.i.d.}}

\newcommand{\ol}{\overline}

\newcommand{\bv}[1]{{\bf#1}}

\newcommand{\Infeas}{\text{Infeas}}

\begin{document}



\title{On Min-Cut Algorithms for \\ Half-Duplex Relay Networks}

\author{
   \IEEEauthorblockN{Ra\'ul Etkin}
   \IEEEauthorblockA{Samsung Information Systems America \\ raul.etkin@samsung.com} 
  \and
   \IEEEauthorblockN{Farzad Parvaresh}
   \IEEEauthorblockA{University of Isfahan \\ f.parvaresh@eng.ui.ac.ir} 
  \and
   \IEEEauthorblockN{Ilan Shomorony and Salman Avestimehr}
   \IEEEauthorblockA{Cornell University \\ is256@cornell.edu, avestimehr@ece.cornell.edu} 
 }

\author{Ra\'ul Etkin, Farzad Parvaresh, Ilan Shomorony and A. Salman Avestimehr%
\thanks{R. Etkin is with Samsung Information Systems America, San Jose, CA 95134, USA. (email: raul.etkin@samsung.com).}%
\thanks{F. Parvaresh is with the department of Electrical Engineering, University of Isfahan, Isfahan, Iran (email: f.parvaresh@eng.ui.ac.ir).}%
\thanks{I. Shomorony and A. S. Avestimehr are with the School of Electrical and Computer Engineering, Cornell University, Ithaca, NY 14853 USA (e-mails: is256@cornell.edu, avestimehr@ece.cornell.edu).}%
\thanks{Part of this paper will appear in the Proceedings of the 2013 IEEE International Symposium on Information Theory \cite{EPSA13}.}}

\maketitle

\begin{abstract}
Computing the cut-set bound in half-duplex relay networks is a challenging optimization problem, since it requires finding the cut-set optimal half-duplex schedule.
This subproblem in general involves an exponential number of variables, since the number of ways to assign each node to either transmitter or receiver mode is exponential in the number of nodes.
We present a general technique that takes advantage of specific structures in the topology of a given network and allows us to reduce the complexity of computing the half-duplex schedule that maximizes the cut-set bound (with i.i.d. input distribution).
In certain classes of network topologies, our approach yields polynomial time algorithms.
We use simulations to show running time improvements over alternative  methods and compare the performance of various half-duplex scheduling approaches in different SNR regimes.
\end{abstract}

\input{intro}


\input{problemsetting}

\input{example}

\input{mainresult}

\input{applications}

\input{extensions}

\input{simulations}

\input{conclusions}

\input{bib}

\end{document}

%% file: intro.tex

\section{Introduction}



Recent years have seen an increasing interest in the benefits that can be obtained from the use of relays in wireless networks.
In the realm of cell phone communications, as the data demands grow at a fast pace, 
the use of techniques such as relays and femto base stations becomes essential to allow for a more efficient utilization of the wireless spectrum.
Furthermore, in wireless sensor networks, by allowing sensor nodes to operate as relays, it is possible to increase the network physical range and reduce the power that must be consumed at each node.
As a result, developing theoretical models of wireless relay networks that are meaningful in a practical context and allow us to study the fundamental limits of communication in relay networks is important for the development of future wireless technologies.

Perhaps one of the biggest discrepancies between the current theory and practice of wireless communications is whether relay nodes are \emph{full-duplex}; i.e., whether they can transmit and receive simultaneously in the same frequency band.
Even though most of the theoretical studies assume all network nodes to be full-duplex, the practical design of RF radios with that capability is very sophisticated and expensive.
Therefore, it is more practical to consider \emph{half-duplex} relay nodes; i.e., nodes that can operate either in transmit mode or in receive mode. 

The main reason for the limited adoption of the half-duplex constraint in theoretical works is the inherent combinatorial aspect it bestows on the problem.
This occurs due to the fact that designing communication schemes for wireless networks with half-duplex nodes requires scheduling the transmit and receive modes of each relay in the network.
Therefore, in order to find capacity-achieving schemes, one must inevitably solve an optimization problem that finds the optimal transmit/receive schedule. 
A simple formulation of this problem involves associating a variable $x_m \in [0,1]$ to each \emph{mode configuration} $m$.
A mode configuration defines which nodes are in transmit mode and which nodes are in receive mode.
We then want to find a schedule, i.e., a set of values for the $x_m$ variables satisfying $\sum_m x_m = 1$, which allows for the highest achievable rate of reliable data transmission.

Two main challenges lie in the way of solving this problem and, thus, characterizing the fundamental limits of communication in half-duplex networks. 
The first one is that the number of $x_m$ variables in the schedule optimization problem is $2^{|V|}$, where $V$ is the set of nodes in the network, and, as we consider larger networks, the problem quickly becomes intractable.
The second challenge is that, even for a fixed half-duplex schedule, the capacity of networks with general topology is unknown. 
This means that we cannot even characterize what a ``good'' schedule is, let alone compute it efficiently.

To overcome the difficulty imposed by the latter challenge, we consider a performance metric that is easier to characterize than the exact capacity. 
One intuitive approach is to look for the schedule 
that maximizes the \emph{cut-set bound}.
Besides being the only known general outer bound for the capacity of relay networks, recent results indicate that it is in fact a very good outer bound in many cases.
As shown in \cite{ADT11}, for linear deterministic wireless relay networks with full-duplex nodes, the capacity is exactly given by the {cut-set bound}, and for full-duplex Gaussian relay networks, the cut-set bound is a contant-gap approximation to the capacity (the value of this gap was later tightened in \cite{SuhasLatticeRelay,NoisyNetworkCoding}).
In addition, \cite{ADT11} also showed that the half-duplex schedule that maximizes the cut-set bound is approximately optimal, in the sense that it can achieve rates within a contant gap of the half-duplex capacity.

The use of the cut-set bound can be further motivated from a computational point of view.
As recently shown in \cite{PE11}, for several network models of interest (such as the linear deterministic model, the Gaussian model and the wireless erasure model \cite{HassibiErasureNet}), the cut-set bound can be efficiently computed\footnote{By restricting the transmit signals of each node to be independent, which is shown in \cite{ADT11} to be optimal for linear deterministic networks and constant-gap-optimal for Gaussian relay networks.}.
Hence, our goal is to study the problem of \emph{efficiently finding the half-duplex schedule that maximizes the cut-set bound}.
The problem of maximizing the half-duplex cut-set bound can be set up as a linear optimization problem whose variables describe the half-duplex schedule.
In general, this problem cannot be solved efficiently as a linear program since it involves an exponential number of variables and constraints.
The techniques introduced in \cite{PE11}, which exploit a submodularity property of mutual information, can be used to solve the optimization problem without evaluating each of the exponentially many cut-set constraints in the linear program.
However, we are still left with an exponential number of variables, reflecting the fact that, in a half-duplex network, the number of mode configurations is exponential in the number of nodes. 
As a result, even if we knew the solution to the optimization problem, the description of the optimal schedule would require an exponentially large amount of memory.
In order to make progress in this direction we ask the following question: is this optimization problem simplified when we consider half-duplex networks with special topology properties?
For example, for the special case of line networks, it was shown in \cite{cheaprelay} that the rate-optimal schedule can be easily computed\footnote{The emphasis in \cite{cheaprelay} is on capacity characterization rather than efficient scheduling. However, the capacity achieving half-duplex schedule for line networks can clearly be efficiently computed.}.
As another example, in \cite{BrahmaHalfDuplex}, for $N$-relay diamond networks, simple half-duplex schedules involving only two states are discussed and are shown to achieve a large fraction of the capacity for $N=2$ and $3$.

Our main result is a technique that allows us to exploit specific
properties of the network topology to reduce the number of
variables in the linear program. 
The main idea is that, if we find (possibly overlapping) subsets of nodes $V_1,V_2,...,V_k$, for some $k$, such that
\begin{itemize}
\item for any cut, the cut value is only a function of the local schedule at each of the $V_i$s,
\item given local schedules for each of the $V_i$s, a global schedule can be constructed,
\end{itemize}
then it is possible to find an equivalent formulation of the optimization problem involving a number of variables that is exponential on $\max_i |V_i|$.
Thus, provided that we can find subsets $V_1,...,V_k$ whose sizes are much smaller than the total number of nodes, the complexity reduction can be very significant.
For example, the line network considered in \cite{cheaprelay} is one special case where $V_1,V_2,...$ can be chosen as all pairs of consecutive nodes ($|V_i|=2$ for all $i$), and the problem of finding the half-duplex schedule that maximizes the cut-set bound (which, in this case, is actually capacity-optimal) can be solved in polynomial time in the size of the network.
The conditions that the $V_i$s need to satisfy will be made precise in Section \ref{mainsec}.

The remainder of the paper is organized as follows. 
In Section \ref{settingsec} we formally define the problem of optimizing the half-duplex schedule to maximize an approximation to the cut-set bound. 
We then introduce the intuition and main ideas used to solve the problem efficiently in Section \ref{examplesec}, by means of a simple example. 
In Section \ref{mainsec}, we use the intuition provided in the example to formalize the properties required from the node subsets $V_1,V_2,...,V_k$.
Then, in Section \ref{equivsec}, we state Theorem \ref{mainthm}, which provides an equivalent formulation of the half-duplex schedule optimization problem, for a given choice of node subsets $V_1,...,V_k$.
Sections \ref{equivproof} and \ref{computproof} contain the proof of Theorem \ref{mainthm}.
In Section \ref{heuristicsec}, we describe a heuristic for finding node subsets $V_1,...,V_k$ satisfying the required properties.
We then proceed to illustrating the usefulness of the result by employing it in two classes of network topologies.
These two classes can be seen as generalizing line networks in two different directions.
First, in Section \ref{layeredsec}, we consider layered networks, where the number of nodes per layer is at most logarithmic in the number of layers.
Then, in Section \ref{linesec}, we consider line networks where, in addition to an edge connecting every pair of consecutive nodes, there is an edge connecting nodes that are at a distance $2$ of each other, thus giving the network a non-layered topology.
For both of these classes, our technique allows the half-duplex schedule that maximizes the cut-set bound to be computed in polynomial time in the network size.
In Section \ref{extsec}, we briefly describe how our techniques can be extended to other network models.
We validate our results in Section \ref{sec:simulations} with a numerical analysis of our proposed algorithm compared to the current best known approach, which requires solving a problem with an exponential number of variables.
We show that our algorithm significantly outperforms the brute force approach in networks of moderate and large size. 
We conclude our paper in Section \ref{sec:conc}.

%% file: problemsetting.tex

\section{Problem setting}
\label{settingsec}

%

In order to simplify the exposition, we will focus on wireless relay networks under the Gaussian model.
Most of the results we present can be easily extended to other models, in particular the linear deterministic model from \cite{ADT11} and the wireless erasure model \cite{HassibiErasureNet}.

A Gaussian relay network is defined by a directed graph $G=(V,E)$ with node set $V$ and edge set $E$, a source node $S \in V$, a destination node $D \in V$, and a real-valued channel gain $h_e$ for each edge $e \in E$.
In this work we will focus in Gaussian relay networks where the nodes are \emph{half-duplex}.
This means that each node in the network cannot simultaneously transmit and receive.
Therefore, a coding scheme for a network with half-duplex nodes must specify, besides the usual encoding, decoding and relaying functions, which nodes transmit and which nodes receive at any given time.
In other words, a coding scheme with blocklength $n$ must define, for each time $t=1,2,...,n$, a partition of the node set $V$ into a set of transmitter nodes $\cT[t]$ and a set of receiver nodes $\cR[t]$.
Thus, at time $t=1,2,...,n$, each transmitter node $v \in \cT[t]$ may transmit any real-valued signal $X_{v}[t]$, whereas a node $v \in \cR[t]$, must transmit $X_v[t] = 0$.
Then, at time $t = 1,2,...,n$, the signal received by a node $v \in \cR[t]$ is given by
\aln{
Y_v[t] = \sum_{u : (u,v) \in E} h_{u,v} X_u[t] + Z_v[t], \text{ for $t=1,2,...,n$},
}
where $Z_v[t]$ is the zero mean unit variance Gaussian discrete-time white noise process associated with node $v$, i.i.d. over time and across nodes. 
If, instead, node $v$ is in $\cT[t]$, its received signal is just $Y_v[t] = 0$.
Notice that we adopt a fixed-schedule half-duplex model, and we do not allow the transmit/receive schedule of a given node to change during communication period based on its received signals. In other words, the half-duplex schedule is independent of the transmitted and received signals.

Communication takes place over a block of $n$ discrete time-steps and follows a coding scheme which specifies the encoding function for the source $S$, decoding function for the destination $D$, and causal relaying functions for all remaining nodes $v \in V \setminus \{S,D\}$.
We assume that the encoding and relaying functions are such that the average power consumed by each node over the $n$ time-steps does not exceed $1$.
Notice that a relay network with more general power constraints (i.e., the transmit power at each node $v$ cannot exceed some $P_v$) is equivalent to another relay network with unity power constraints where the specific power constraints are absorbed into the outgoing channel gains of a given a node.
For a coding scheme with rate $\log M$, source $S$ picks a message $W \in \{1,...,M\}$ that it wishes to communicate to $D$, and transmits signals $X_{S}[t]$, $t=1,...,n$, according to its encoding function.
The destination uses a decoding function to obtain an estimate $\hat W \in \{1,...,M\}$ of $W$  
from the $n$ received signals to the source message indices.
We say that rate $\log M$ is achievable if the probability of error in the decoding of $W$ can be made arbitrarily close to $0$ by choosing a sufficiently large $n$. 
The capacity $C$ is the supremum of the achievable rates. 

While the exact  characterization of the capacity of Gaussian relay networks is a long-standing open problem even in the full-duplex case, it has been approximated to within a constant number of bits that does not depend on channel gain values \cite{ADT11,SuhasLatticeRelay,NoisyNetworkCoding}.
As shown in \cite{ADT11}, a coding scheme known as quantize-map-and-forward can achieve to within this constant number of bits of the cut-set bound.
Furthermore, it is shown that, if the cut-set bound is computed by restricting the joint distribution on the nodes' transmit signals to an i.i.d.~distribution where each node transmits an independent zero-mean unit-variance Gaussian signal, 
it remains within a constant gap of the actual cut-set bound.

The half-duplex scenario is also explored in \cite{ADT11},
where an approximation of the half-duplex capacity is established using the important
concept of a \emph{mode configuration}.
Since at each time $t$, each node $v \in V$ must be either in transmitter mode or in receiver mode, the modes of all nodes in the network at a given time $t$ can be described by a binary $|V|$-dimensional vector $\bv m$, the mode configuration vector.
Moreover, for a given coding scheme we can let $q(\bv m)$ be the fraction of the $n$ time-steps where the network uses the mode configuration described by $\bv m$.
Thus $q$ is essentially a probability mass function over the $2^{|V|}$ mode configuration vectors.
For a mode configuration $\bv m$, we will let $\cT(\bv m)$ be its transmitter nodes and $\cR(\bv m)$ be its receiver nodes.
We formally define $\cT(\bv m) = \{ v \in V \,: \, m_v = 0\}$ and $\cR(\bv m) = \{ v \in V \,: \, m_v = 1\}$.
For Gaussian relay networks with half-duplex nodes, \cite{ADT11} establishes that a constant-gap capacity approximation is given by
\al{ \label{ciihd}
\ciid = \max_{q} \min_{\Omega} \sum_{\bv m} q(\bv m) I\left(X_{\Omega \cap \cT(\bv m)} ; Y_{\Omega^c \cap \cR(\bv m)} | X_{\left(\Omega\cap \cT(\bv m)\right)^c} \right),
}
where the transmit signals $X_i$, $i \in V$, are independent and distributed as $\N(0,1)$.
The quantity $\ol C_{\iid}$ is essential to the results presented in this paper.
The technique we introduce for finding a half-duplex schedule can be described as finding the distribution $q$ which attains the maximum in \Eref{ciihd}.
We point out that, for any fixed mode configuration distribution $q$, the results in \cite{PE11} allow 
\al{ \label{mincut}
\min_{\Omega} \sum_{\bv m} q(\bv m) I\left(X_{\Omega \cap \cT(\bv m)} ; Y_{\Omega^c \cap \cR(\bv m)} | X_{\left(\Omega\cap \cT(\bv m)\right)^c} \right)
}
to be computed efficiently.
In effect, the authors show that, when the the transmit signals $X_i$, $i \in V$, are i.i.d.~Gaussian random variables, the conditional mutual information $I(X_{\Omega}; Y_{\Omega^c} | X_{\Omega^c})$ is a submodular function of $\Omega \in 2^V$.
Thus, the minimization in \Eref{mincut} can be solved efficiently (i.e., in polynomial time in the input network size).

Nonetheless, when we consider the computation $\ol C_{\iid}$, the maximization over mode configuration distributions $q(\bv m)$ adds an extra layer of difficulty to the problem.
This is easily seen by noticing that $q(\bv m)$ is a probability distribution on $2^{|V|}$ elements, rendering the brute force approach computationally inefficient.
This approach for computing $\ol C_{\iid}$ can be formally described through the following linear program, which we refer to as Problem 1.

\vspace{4mm}

{\noindent \bf Problem 1:}
\begin{align}
\label{eq:optim_prob1}
\underset{R, q(\bv m)}{\text{maximize}} \  &  R  \\ \nonumber
\text{subject to } & \\ 
\label{prob1:c0}
	& R < \sum_{\bv m \in \{0,1\}^{|V|} } 
			q\left( \mathbf{m}\right)
			I\left(X_{\Omega \cap \cT(\mathbf{m})} ; Y_{\Omega^c \cap \cR(\mathbf{m})}|X_{(\Omega \cap \cT(\mathbf{m}))^c}\right), 
				\ \;  \forall \, \Omega \in 2^V: S \in \Omega, D \not\in \Omega\\
\label{prob1:c1}
	& 0 \le q\left( \bv m \right), \forall \, \bv m \in \{0,1\}^{|V|} \\
\label{prob1:c2}
	& \sum_{ \bv m \in \{0,1\}^{|V|} } 
		q\left( \bv m \right) = 1 \\
&\text{$X_i$, $i \in V$, are independent $\N(0,1)$ random variables} \nonumber
\end{align}
\begin{rmk}
We point out that in Problem 1 there is no optimization of the power allocation of the nodes over the various modes. While power optimization may lead to a smaller gap to the cut-set bound it would also destroy the convex structure of the problem. 
In addition, the analysis of channels subject to the half-duplex constraint is more meaningful in the high SNR regime, where the channel capacity is degrees-of-freedom limited. In this regime optimization over the power allocations would not provide significant performance improvements. 
\end{rmk}

Notice that, although Problem 1 is defined as a linear program, it involves $O(2^{|V|})$ constraints and variables and, therefore, cannot be solved efficiently.
Moreover, we emphasize that even though Problem 1 is set up as a maximization over $R$, it is not actually computing achievable rates.
In fact, its main outcome is the maximizing distribution $q$, and Problem 1 should be understood as a tool to compute a ``good'' half-duplex schedule for a wireless relay network.

Problem 1 can be generalized to find ``good'' half-duplex schedules for a wireless relay network where the goal is not only to maximize the achievable rate, but also to minimize the duty cycle of the nodes, which can be related to the power consumption of the network. This generalization is given by the following linear program,
\vspace{4mm}

{\noindent \bf Problem 2:}
\begin{align}
\label{eq:optim_prob2}
\underset{R, T_{tot}, q(\bv m)}{\text{minimize}} \  &  \mu_1 R + \mu_2 T_{tot} \\ \nonumber
\text{subject to } & \\ 
\label{prob2:c0}
	& R < \sum_{\bv m \in \{0,1\}^{|V|} } 
			q\left( \mathbf{m}\right)
			I\left(X_{\Omega \cap \cT(\mathbf{m})} ; Y_{\Omega^c \cap \cR(\mathbf{m})}|X_{(\Omega \cap \cT(\mathbf{m}))^c}\right), 
				\ \;  \forall \, \Omega \in 2^V: S \in \Omega, D \not\in \Omega\\
\label{prob2:c1}
	& 0 \le q\left( \bv m \right), \forall \, \bv m \in \{0,1\}^{|V|} \\
\label{prob2:c2}
	& \sum_{ \bv m \in \{0,1\}^{|V|} } 
		q\left( \bv m \right) = 1 \\
\label{prob2:c3}
&  \sum_{i=1}^{|V|} \sum_{\bm \in \{0,1\}^{|V|}:m_i=1} q(\bm) \le T_{tot}  \\
\label{prob2:c4}
&	C_{min} \le R \\
&\text{$X_i$, $i \in V$, are independent $\N(0,1)$ random variables} \nonumber
\end{align}
for fixed constants $\mu_1$, $\mu_2$, and $C_{min}$, which can be chosen appropriately to find good schedules for different objectives. For example, setting $\mu_1=-1$, $\mu_2=0$, $C_{min}=0$, Problem 2 maximizes the rate subject to the half duplex constraint and reduces to Problem 1. Alternatively, setting $\mu_1=0$, and $\mu_2=1$, for some $C_{min} > 0$ Problem 2 minimizes the sum of the duty cycles of the relay nodes subject to a minimum rate contraint for the half duplex network.

Our main goal will be to find equivalent formulations of Problem 2 which can be efficiently computed.
We will show that special structures in the network topology can be exploited in order to replace constraints (\ref{prob2:c0}), (\ref{prob2:c1}) and (\ref{prob2:c2}) with equivalent ones and to reduce the number of variables, which leads, in some cases, to polynomial-time algorithms for solving Problem 2.

%% file: example.tex

\section{A simple example} \label{examplesec}
We will illustrate the main ideas presented in this paper through a simple example. We consider a single source single destination Gaussian relay network with five relays, as shown in Figure \ref{fig:network_1}. 
\begin{figure}[htb] \center
		\includegraphics[height=4cm]{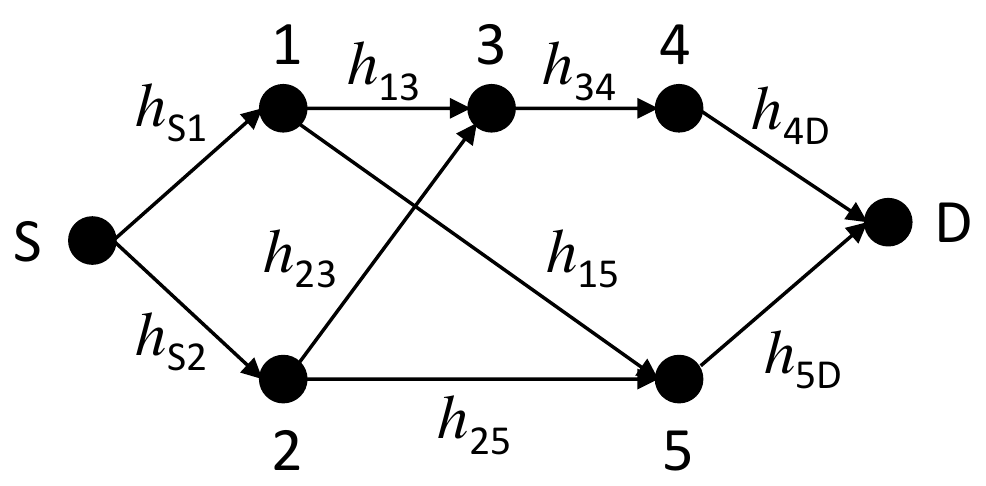}
	\caption{A Gaussian relay network with 5 relays.}
	\label{fig:network_1}
\end{figure}
The set of nodes in this network is $V=\{S, 1, 2, 3, 4, 5, D \}$. 
As explained in the previous section, the half-duplex capacity of this network can be characterized within a constant gap by computing the cut-set bound with independent Gaussian encoding at the various nodes. 
Therefore, the schedule $q$ computed via Problem 1 is approximately optimal.

The computational complexity in solving Problem 1 for the network in Figure \ref{fig:network_1} arises from the $2^5=32$ inequalities of the form (\ref{prob1:c0}) involving all the cuts in the network, and the $2^5=32$ scalar variables required to represent the distribution $q(\bm)$ of the half-duplex schedules. 
We would like to find a more efficient approach for solving this optimization. 
As shown in \cite{PE11}, it is possible to exploit a submodularity property of mutual information to simplify the computation of the min cut, and as a result, it is not necessary to compute (\ref{prob1:c0}) for each of the 32 cuts. 
To further simplify the computation we will show that it is possible to reduce the number of optimization variables represented by $q(\bm)$.
\begin{figure}[htb] \center
		\includegraphics[height=6cm]{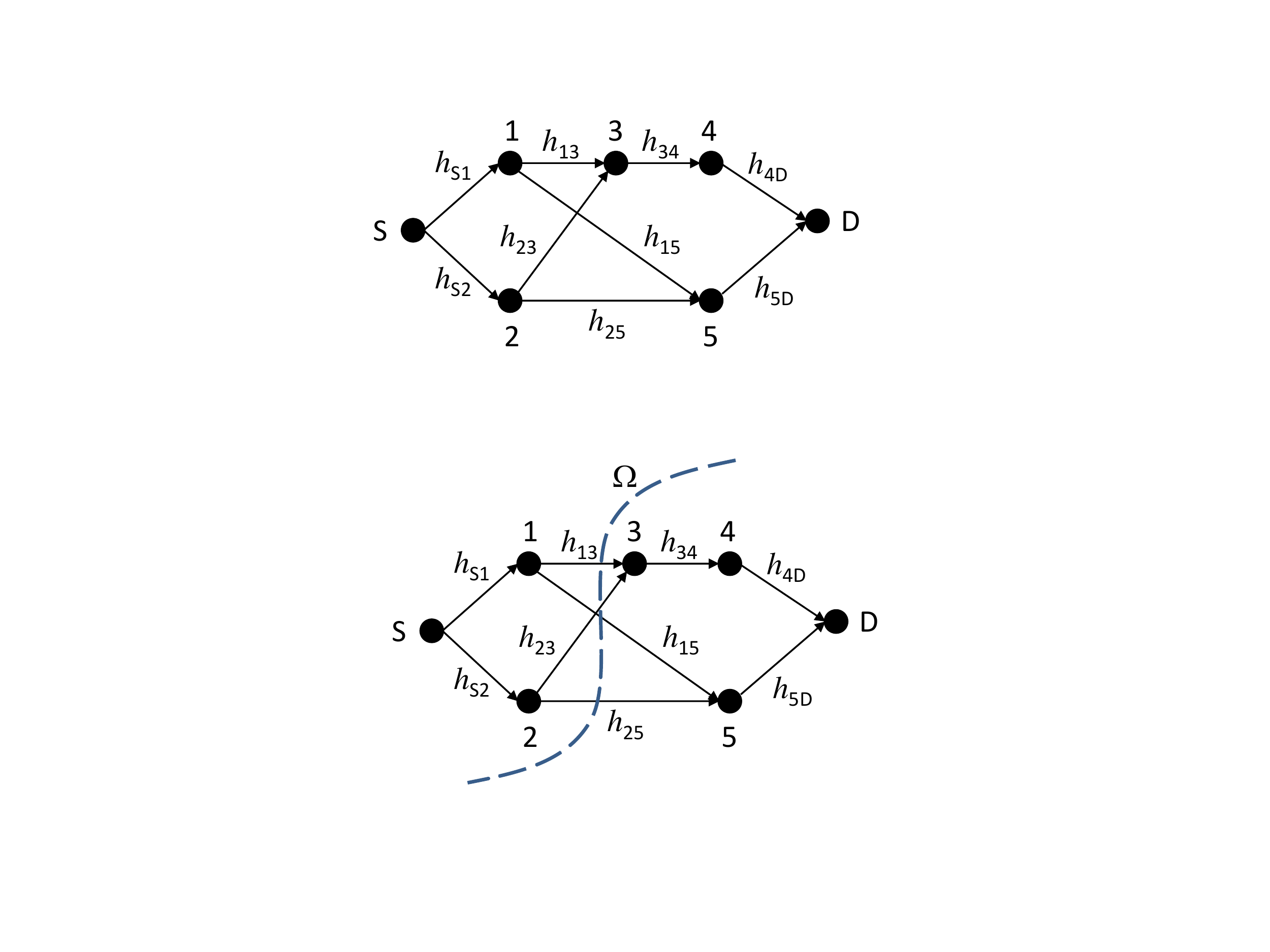}
	\caption{A cut $\Omega$ in the network.}
	\label{fig:cut_1}
\end{figure}
Consider the cut $\Omega=\{S, 1, 2\}$ as shown in Figure \ref{fig:cut_1}, and the corresponding inequality of the form (\ref{prob1:c0})
\[
R < \sum_{\mathbf{m}\in \{0,1\}^5 } 
			q\left( \mathbf{m}\right)
			I\left(X_{\{S, 1, 2\} \cap \cT(\mathbf{m})} ; Y_{\{3, 4, 5, D\} \cap \cR(\mathbf{m})}|X_{(\{S, 1, 2\} \cap \cT(\mathbf{m}))^c}\right). 
\]
Due to the network connectivity, the capacity of this cut only depends on the variables (input, output, and transmit/receive mode) of the nodes 1, 2, 3, and 5. As a result, the inequality can be rewritten as
\begin{align*}
R <& \sum_{\mathbf{m}\in \{0,1\}^5 } 
			q\left( \mathbf{m}\right)
			I\left(X_{\{1, 2\} \cap \cT_{12}(m_1,m_2)} ; Y_{\{3, 5\} \cap \cR_{35}(m_3, m_5)}|X_{ \{1,2\} \cap \cR_{12}(m_1,m_2))}\right)
			\\
=&	\sum_{(m_1,m_2,m_3,m_5)\in \{0,1\}^4 } 
			\left[\sum_{m_4 \in \{0,1\}} q(m_1,\ldots, m_5)\right]
			I\left(X_{\{1, 2\} \cap \cT_{12}(m_1,m_2)} ; Y_{\{3, 5\} \cap \cR_{35}(m_3, m_5)}|X_{ \{1,2\} \cap \cR_{12}(m_1,m_2)}\right)
			\\
=&	\sum_{(m_1,m_2,m_3,m_5)\in \{0,1\}^4 } 
			q_{1235}(m_1, m_2, m_3, m_5)
			I\left(X_{\{1, 2\} \cap \cT_{12}(m_1,m_2)} ; Y_{\{3, 5\} \cap \cR_{35}(m_3, m_5)}|X_{ \{1,2\} \cap \cR_{12}(m_1,m_2)}\right)
\end{align*}
where we denoted by $\cT_{i_1,\ldots,i_k}(m_{i_1},\ldots, m_{i_k})=\cT(\bm) \cap \{i_1,\ldots,  i_k\}$ and  $\cR_{i_1,\ldots,i_k}(m_{i_1},\ldots, m_{i_k})= \cR(\bm) \cap \{i_1,\ldots, i_k\}$,  and defined $q_{1235}(m_1, m_2, m_3, m_5)=\sum_{m_4 \in \{0,1\}} q(m_1,\ldots, m_5)$.
\begin{figure}[htb] \center
		\includegraphics[height=6cm]{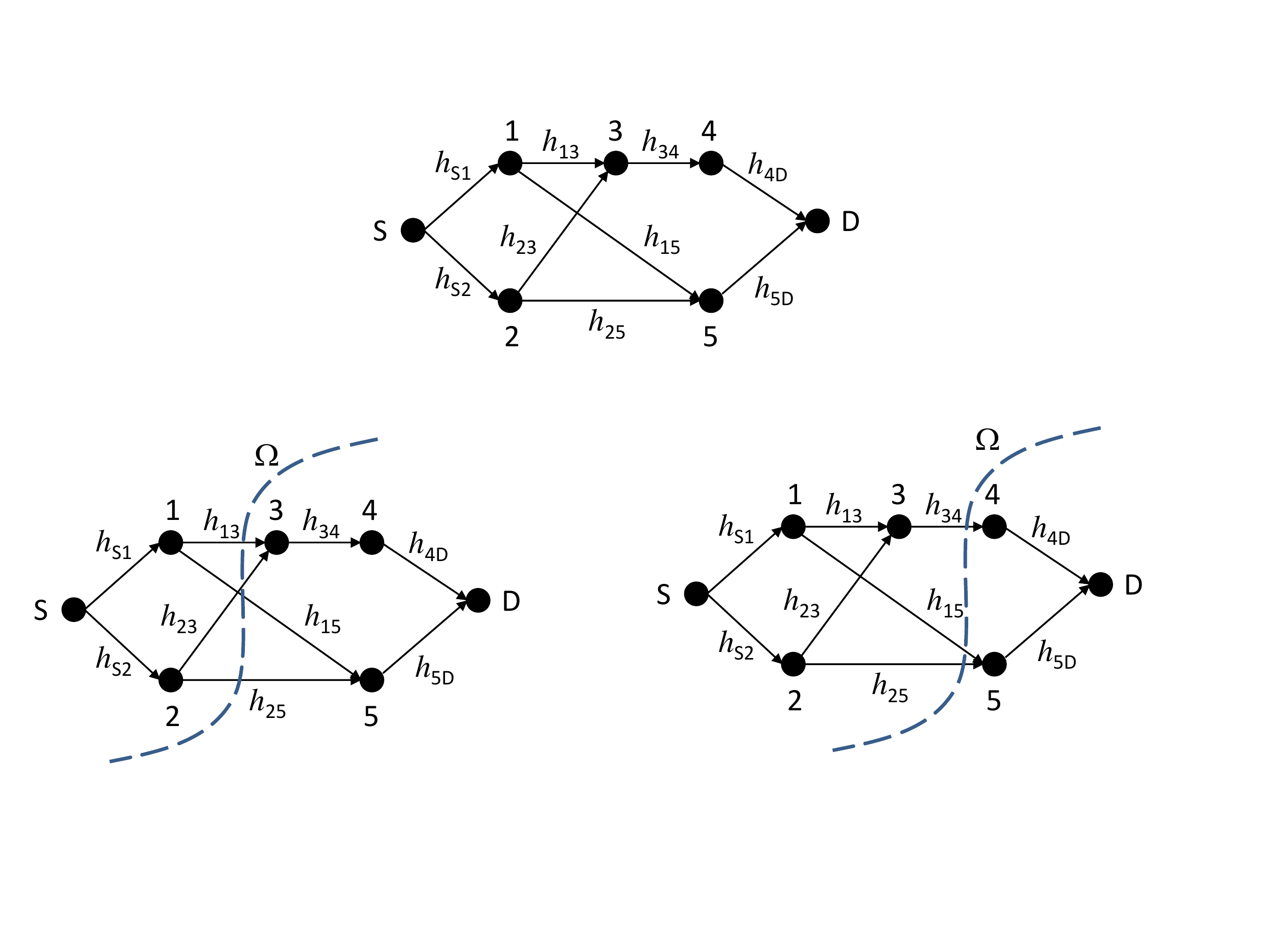}
	\caption{Another cut $\Omega$ in the network.}
	\label{fig:cut_2}
\end{figure}
In a similar way, consider the cut $\Omega=\{S, 1, 2, 3\}$ as shown in Figure \ref{fig:cut_2}. 
Proceeding as before, we observe that the rate inequality (\ref{prob1:c0}) can be rewritten as
\begin{align*}
R <	&\sum_{(m_1,m_2, m_5)\in \{0,1\}^3 } 
			q_{125}(m_1, m_2, m_5)
			I\left(X_{\{1, 2\} \cap \cT_{12}(m_1,m_2)} ; Y_{\{5\} \cap \cR_5(m_5)}|X_{\{1,2\} \cap \cR_{12}(m_1,m_2)}\right)
\\
&+ \sum_{(m_3, m_4)\in \{0,1\}^2 } 
			q_{34}(m_3, m_4)
			I\left(X_{\{3\} \cap \cT_3(m_3)} ; Y_{\{4\} \cap \cR_4(m_4)}\right)
\end{align*}
where we denoted by $q_{125}$ and $q_{34}$ the corresponding marginals of $q$.

By considering each of the 32 cuts, we observe that all the cut capacities can be decomposed into terms that are completely determined by the marginals $q_{1235}$ and $q_{34}$. Since $q_{1235}$ and $q_{34}$ must correspond to some joint distribution $q$, they must satisfy 
\[
\sum_{(m_1,m_2,m_5)\in \{0,1\}^3} q_{1235}(m_1,m_2,m_3,m_5) = \sum_{m_4 \in\{0,1\}} q_{34}(m_3,m_4).
\]
Assuming that this consistency requirement is satisfied, we can find a joint distribution $\tilde{q}$ with these marginals:
\[
\tilde{q}(m_1,\ldots,m_5)=\frac{q_{1235}(m_1,m_2,m_3,m_5)\cdot q_{34}(m_3,m_4)}{q_3(m_3)}
\]
where $q_3(m_3)=\sum_{m_4 \in\{0,1\}} q_{34}(m_3,m_4)$. 
While $\tilde{q}$ may be different from $q$, they both achieve the same rate in (\ref{eq:optim_prob1}). 
This can be seen by noting that both distributions result in the same cut capacities. 
As a result, we can simplify the Problem 1 obtaining
\begin{align}
\label{eq:optim_example_2}
\underset{R, q_{1235}, q_{34}}{\text{maximize}} &\    R  \\ \nonumber
\text{subject to }&  \\ 
\label{optex:c0_2}
	& R < \sum_{\mathbf{m}\in \{0,1\}^4 } 
			q_{1235}(\mathbf{m})
			I\left(X_{(\Omega \backslash \{3\}) \cap \cT_{1235}(\mathbf{m})} ; Y_{\Omega^c \cap \cR_{1235}(\mathbf{m})}|X_{((\Omega \backslash \{3\})\cap \cT_{1235}(\mathbf{m}))^c}\right)
\\
&\quad \quad \quad  + \sum_{\mathbf{m}\in \{0,1\}^2 } 
			q_{34}(\bm)	I\left(X_{\Omega \cap \{3\} \cap \cT_{34}(\bm)} ; Y_{\Omega^c \cap \{4\} \cap \cR_{34}(\bm)} \right), 
\quad \forall \; \Omega \in 2^V: S \in \Omega, D \not\in \Omega \nonumber \\
\label{optex:c1_2}
	& 0 \le q_{1235}\left(m_1,m_2,m_3,m_5\right), \text{ for all } (m_1,m_2,m_3,m_5) \in \{0,1\}^4 \\
\label{optex:c1_3}
	& 0 \le q_{34}\left(m_3,m_4\right), \text{ for all } (m_3,m_4) \in \{0,1\}^2 \\
\label{optex:c2_2}
	& \sum_{ (m_1, m_2, m_3, m_5) \in \{0,1\}^4 } 
		q_{1235}\left(m_1, m_2, m_3, m_5\right) = 1 \\
\label{optex:c2_3}
	& \sum_{ (m_3 m_4) \in \{0,1\}^2 } 
		q_{34}\left(m_3, m_4\right) = 1 \\
\label{optex:c3}
	& \sum_{(m_1,m_2,m_5)\in \{0,1\}^3} q_{1235}(m_1,m_2,m_3,m_5) = \sum_{m_4 \in\{0,1\}} q_{34}(m_3,m_4), \text{ for all } m_3 \in \{0,1\}.
\end{align}

The main simplification in the optimization problem above arises from the smaller number of variables in $q_{1235}$ and $q_{34}$, with a total of $2^4+2^2=20$, as compared to $2^5=32$ appearing in $q$ in (\ref{eq:optim_prob1}).
While in this simple network the computational complexity savings are small, this approach can lead to substantial simplification of the problem in large networks.



%% file: mainresult.tex

\section{Finding Half-Duplex Schedules in General Relay networks} \label{mainsec}

The main contribution of this paper is in developing techniques to more efficiently solve Problem 2.
Based on the intuition provided by the example in the previous section, we will show that, given a \emph{node grouping} satisfying some properties, the computational complexity of Problem 2 can be significantly improved.
By a node grouping we simply mean a list of (possibly overlapping) subsets $V_1, ..., V_k$ of the nodes $V$ such that $\cup_i V_i = V$.

In this section, we will state and prove our main results.
We will first show that, for a given node grouping $V_1,...,V_k$, Problem 2 can be re-written as another optimization problem, which we call Problem 3, with a smaller number of variables.
This result is stated in Subsection \ref{equivsec} and proved in Subsections \ref{equivproof}  and \ref{computproof}.
Finally, in Subsection \ref{heuristicsec}, we discuss the problem of finding a good node grouping and we provide a heuristic that finds a valid node grouping (as defined in Subsection \ref{equivsec}) for an arbitrary network.

\subsection{Finding an Equivalent Optimization Problem Given a Node Grouping} \label{equivsec}

In order to generalize the ideas shown in the example in section \ref{examplesec} to other networks, we identify the main elements that led to the simplification of the optimization problem.
\begin{itemize}
\item The cut capacities for all cuts depend on marginals $q_{V_1},\ldots, q_{V_k}$ of $q$ of (significantly) smaller dimension, i.e. $\max_i |V_i| \ll |V|$, where $V$ is the set of nodes and $V_i$s are the subsets of nodes corresponding to the marginals. 
\item Given an arbitrary set of probability distributions $\{q_{V_1},\ldots, q_{V_k}\}$ satisfying some consistency constraints, there exists a joint distribution with these probability distributions as marginals. 
\end{itemize}
In order to characterize when the cut capacities only depend on marginals $q_{V_1},\ldots, q_{V_k}$, we consider defining, for a given cut $\Omega$, the cut graph $G_\Omega=(V,E_\Omega)$ with $E_\Omega=\{(i,j)\in E: i\in \Omega, j \in \Omega^c\}$. 
That is, $G_\Omega$ is the subgraph of $G$ obtained by keeping only the arcs going from nodes in $\Omega$ to nodes in $\Omega^c$. 
In addition, let $G_{\Omega,i}=(V_{\Omega,i},E_{\Omega,i})$ be the $i$th connected component of $G_\Omega$, for $i=1,\ldots, N(\Omega)$, where $N(\Omega)$ is the number of such connected components. 

The consistency constraints in the second bullet can be formalized by requiring that, for two sets $V_i$ and $V_j$ and respective probability distributions $q_{i}$ and $q_{j}$ defined over these sets, their marginals corresponding to the node set $V_i \cap V_j$ are the same.
Thus, we are interested in a node grouping $V_1, V_2, \ldots, V_k \subseteq V$ satisfying
\begin{enumerate}[P1. ]
\item $V_{\Omega,i} \subseteq V_j$ for some $j \in \{1,\ldots, k\}$, for all connected components $V_{\Omega,i}$, for all cuts $\Omega$.
\item Given an arbitrary set of consistent distributions $\{q_{1},\ldots, q_{k}\}$, there exists a joint distribution $\tilde q$ with these distributions as marginals.
\end{enumerate}
\begin{rmk}
The number of distributions $k$ will be assumed to be smaller than or equal to the number of nodes in the network $|V|$. This assumption will be justified in Section \ref{heuristicsec}.
\end{rmk}
Now, for a general network $G$ and node subsets $V_1,...,V_k$ satisfying properties P1 and P2, by letting $\cT(\Omega,i,\mathbf{m})=\{v \in V_{\Omega,i}:  m_v = 1\}$, $\cR(\Omega,i,\mathbf{m})=\{v \in V_{\Omega,i} :  m_v = 0\}$, for any $U \subseteq V$,
$r(\Omega,i)$ be a function that returns the index $j$ of a set $V_j$ containing $V_{\Omega, i}$, $s(i)$ be a function that returns the index $j$ of a set $V_j$ containing node $i$, and $t(i)$ be a function that returns the coordinate index in the vector $\bm$ associated to the transmit/receive mode of node $i$ in $q_{s(i)}(\bm)$,
we can define the following optimization problem,

\vspace{3mm}

{\noindent \bf Problem 3:}
\begin{align}
\label{eq:optim_prob3}
\underset{R, T_{tot}, \{q_i\}_{i=1}^k}{\text{minimize}}& \  \mu_1 R +\mu_2 T_{tot}   \\
\text{subject to}  &   \nonumber \\
\label{prob3:c0}
&   R < \sum_{i=1}^{N(\Omega)} \sum_{\mathbf{m}\in \{0,1\}^{|V_{r(\Omega,i)}|  }} 
	q_{r(\Omega,i)}\left( \mathbf{m} \right) 
		\cdot 
	I\left( X_{\Omega \cap \cT(\Omega,i,\mathbf{m})} ; Y_{\Omega^c \cap \cR(\Omega,i,\mathbf{m})} | X_{(\Omega \cap \cT(\Omega,i,\mathbf{m}))^c} \right), \nonumber \\
	& \quad \quad \quad  \forall \, \Omega \in 2^V: S \in \Omega, D \not\in \Omega, \\
\label{prob3:c1}
	&  0 \le q_i\left(\mathbf{m}\right), \mathbf{m}\in \{0,1\}^{|V_i|},  
		\ \text{ for } i=1,2 \ldots, k \\
\label{prob3:c2}
	&  \sum_{\mathbf{m}\in \{0,1\}^{|V_i|}} 
		q_i\left(\mathbf{m}\right) = 1,
			\ \text{ for } i=1,2, \ldots, k \\
\label{prob3:c3}
	&  \sum_{\mathbf{m}\in \{0,1\}^{|V_i|} :\, (m_j)_{j\in V_i \cap V_l}=\mathbf{m}_1} 
			q_{i}\left(\mathbf{m}\right) 
		=
		\sum_{\mathbf{m}\in \{0,1\}^{|V_l|} :\, (m_j)_{j\in V_i \cap V_l}=\mathbf{m}_1} 
			q_{l}\left(\mathbf{m}\right)  \\
& \quad \quad \text{ for all } \mathbf{m}_1 \in \{0, 1\}^{|V_i \cap V_l|}, \text{ for } i\ne l \text{ and } i, l \in \{1,2, \ldots, k\}, \nonumber\\
\label{prob3:c4}
	&  \sum_{i=1}^{|V|} \sum_{\bm \in \{0,1\}^{|V_{s(i)}|}:m_{t(i)}=1} q_{s(i)}(\bm) \le T_{tot} \\
\label{prob3:c5}
&	C_{min} \le R \\
&\text{$X_i$, $i \in V$, are independent $\N(0,1)$ random variables} \nonumber
\end{align}
for fixed constants $\mu_1$, $\mu_2$, and $C_0$, which can be chosen appropriately to find good schedules for different objectives. 

The following theorem establishes that Problem 3 can be seen as a simplification of Problem 2.

\begin{thm} 
\label{mainthm}
For a Gaussian relay network given by graph $G=(V,E)$ with node subsets $V_1,...,V_k$ satisfying P1 and P2, the optimal solutions of Problems 2 and 3 are the same.
Moreover, Problem 3 has $O(|V| 2^{\max_i |V_i|})$ variables (as opposed to the $O(2^{|V|})$ variables in Problem 2), and the optimization can be solved in a time complexity that is polynomial in $|V| 2^{\max_i |V_i|}$.
\end{thm}

\begin{rmk} \Tref{mainthm} can be specialized to simply maximize $\ciid$ by choosing $\mu_1=-1$, $\mu_2=0$ and $C_0=0$.
Similarly, it can specialized to minimize the sum of the duty cycles of the nodes to achieve a minimum value of $\ciid$, by choosing $\mu_1=0$, $\mu_2=1$ and $C_0>0$.
\end{rmk}
%
%

\vspace{2mm}

In the following two subsections we prove Theorem \ref{mainthm}.
First we prove that the solutions of Problems 2 and 3 are equal.
Then we prove that Problem 3 can be solved in a time complexity that is polynomial in $|V| 2^{\max_i |V_i|}$.

\subsection{Proof of the equivalence of Problems 2 and 3}
\label{equivproof}

For any cut $\Omega$ that includes the source and excludes the destination, under the assumption of independent encoding, the cut capacity $I(X_\Omega; Y_{\Omega^c} | X_{\Omega^c})$ can be computed over the subgraph $G_\Omega$, which we denote by $I_{G_\Omega}(X_\Omega; Y_{\Omega^c} | X_{\Omega^c})$, by setting $h_{i,j}=0$ when $(i,j) \not \in E_{\Omega,i}$. 
In this subgraph, with independent encoding, we have $I_{G_\Omega}(X_\Omega; Y_{\Omega^c} | X_{\Omega^c})=I_{G_\Omega}(X_\Omega; Y_{\Omega^c})$. 
In addition, since $G_\Omega$ has connected components $\{G_{\Omega,i}\}_i$, we have $I_{G_\Omega}(X_\Omega; Y_{\Omega^c})=\sum_i I_{G_{\Omega,i}}(X_{\Omega}; Y_{\Omega^c })=\sum_i I(X_{\Omega\cap V_{\Omega,i}}; Y_{\Omega^c \cap V_{\Omega,i}}|X_{\Omega^c})$.
Therefore, we can decompose the cut capacity of any cut $\Omega$ that includes the source and excludes the destination for any mode configuration $\bm$ as follows:
\begin{equation}
I\left(X_{\Omega\cap \cT(\bm)}; Y_{\Omega^c \cap \cR(\bm)} | X_{(\Omega\cap \cT(\bm))^c}\right)  
=\sum_{i=1}^{N(\Omega)} I\left(X_{\Omega \cap V_{\Omega,i} \cap \cT(\bm)}; Y_{\Omega^c \cap V_{\Omega,i} \cap \cR(\bm)} | X_{(\Omega \cap V_{\Omega,i} \cap \cT(\bm))^c}\right).
\label{eq:I_split}
\end{equation}
where we added conditioning to eliminate the influence of the transmit variables $X_i$ associated to nodes in receive mode.
Note that in \eq{eq:I_split} for given cut $\Omega$ and connected component $i$ the corresponding mutual information only depends on the mode configuration of the nodes in $V_{\Omega, i}$. 
As a result, we can write
\begin{align}
\sum_{\bm \in \{ 0, 1 \}^{|V|} }  &
	q\left(\bm \right) 
	I\left(X_{\Omega\cap \cT(\bm)}; Y_{\Omega^c \cap \cR(\bm)} | X_{(\Omega\cap \cT(\bm))^c}\right)  =   \nonumber \\
& = \sum_{ \bm \in \{ 0,1 \}^{|V|} } 
	q(\bm) 
	\sum_{i=1}^{N(\Omega)} I\left(X_{\Omega \cap V_{\Omega,i} \cap \cT(\bm)}; Y_{\Omega^c \cap V_{\Omega,i} \cap \cR(\bm)} | X_{(\Omega \cap V_{\Omega,i} \cap \cT(\bm))^c}\right) \nonumber \\
& = \sum_{i=1}^{N(\Omega)} \sum_{ \bm \in \{ 0,1 \}^{|V|} } 
		q(\bm) 
		I\left(X_{\Omega \cap V_{\Omega,i} \cap \cT(\bm)}; Y_{\Omega^c \cap V_{\Omega,i} \cap \cR(\bm)} | X_{(\Omega \cap V_{\Omega,i} \cap \cT(\bm))^c}\right) \nonumber \\
& =  \sum_{i=1}^{N(\Omega)} \sum_{\mathbf{m}\in \{0,1\}^{|V_{r(\Omega,i)}|  }} 
	q_{r(\Omega,i)}\left( \mathbf{m} \right) 
		\cdot 
	I\left( X_{\Omega \cap \cT(\Omega,i,\mathbf{m})} ; Y_{\Omega \cap \cR(\Omega,i,\mathbf{m})} | X_{(\Omega \cap \cT(\Omega,i,\mathbf{m}))^c} \right).
\label{eq:equivalents_of_I_2}
\end{align}

Equality \eq{eq:equivalents_of_I_2} shows that condition \eq{prob3:c0} is equivalent to condition \eq{prob2:c0}.
Thus in order to compute the cut-set capacity for any cut $\Omega$ we only need to know the marginals $q_1,...,q_k$ of the probability distribution $q$.

Next, we observe that any probability distribution $q(\bm)$ which satisfies  conditions \eq{prob2:c1} and \eq{prob2:c2} also satisfies conditions 
\eq{prob3:c1}, \eq{prob3:c2} and \eq{prob3:c3}, where $q_i(\bm_1) = \sum_{\bm \in \{0,1\}^n :\, (m_j)_{j \in V_i }=\bm_1} q\left(\bm \right)$, for $\bm_1 \in \{0,1\}^{|V_i|}$, and $i=1,\ldots, k$. 

As a result, the solution to Problem 3 is no larger than the solution to Problem 2.
In addition, any set of marginals $q_i(\cdot)$, $i=1,\ldots, k$ that satisfy conditions \eq{prob3:c1}, \eq{prob3:c2} and \eq{prob3:c3} are consistent, and, by property P2, there exists a joint distribution $\tilde q$ having the $q_i$s as marginals.
As a result, any feasible solution in Problem 2 is a feasible solution in Problem 3.
Therefore, the solution to Problem 2 is no larger than the solution to Problem 3, and it follows that they have the same solution.

The variables in Problem 3 correspond to the the values of $q_i(\bm)$ for each $\bm \in \{0,1\}^{|V_i|}$ and each $i = 1,...,k$, in addition to $R$ and $T_{tot}$, and thus the number of variables is $2+\sum_{i=1}^k 2^{|V_i|} = O(|V| \cdot 2^{\max_i |V_i|})$.

\subsection{Time Complexity of Solving Problem 3} \label{computproof}

The computability of the solution in polynomial time in $|V| 2^{\max_i |V_i|}$ will be shown by extending the results of \cite{PE11}. We will use the following definitions and result. The reader is referred to \cite{ConvOpt} for more details.

\begin{definition}[Polynomial computability] 
\label{def:1}
A family of optimization programs is polynomially computable if:
\begin{itemize} 
\item [(i)] for any instance of the program and any point $\bx$ in the domain, the objective and its subgradient can be computed in polynomial time in the size of the instance. 
\item [(ii)] for a given measure of infeasibility $\Infeas(\cdot)$, it should be possible to determine if $\Infeas(\bx)\le \ep$ in polynomial time, and when this inequality is not satisfied, it should be possible to find in polynomial time a vector $\bc$ such that
\[
\bc^T \bx > \bc^T \by, \forall \by: \Infeas(\by) \le \ep.
\]
\end{itemize}
\end{definition}

\begin{definition}[Polynomial growth]
\label{def:2}
A family of optimization programs has polynomial growth if the objectives and the infeasibility measures as functions of points $\bx$ in the domain grow polynomially with $\|\bx\|_1$.
\end{definition}

\begin{definition}[Polynomial boundedness of feasible sets]
\label{def:3}
A family of optimization programs has polynomially bounded feasible sets if the feasible set of an instance of the program is contained in an Euclidean ball centered at the origin with radius that grows at most polynomially with the size of the instance.
\end{definition}

\begin{prop}[{\cite[Theorem 5.3.1]{ConvOpt}}]
\label{prop:convopt}
Let $\cP$ be a family of convex optimization programs equipped with infeasibility measure $\Infeas(\cdot)$. Assume that the family is polynomially computable with polynomial growth and with polynomially bounded feasible sets. Then $\cP$ is polynomially solvable.
\end{prop}

In order to use \Pref{prop:convopt} we need to check that the optimization \eq{eq:optim_prob3} is a convex program. Since the objective function is linear, we only need to check that the feasible set, which we denote by $\mathcal{K}$, is convex.  

\begin{lem}
\label{lem:K_convex}
The feasible set $\mathcal{K}$ is a convex set.
\end{lem}

\begin{proof}(Lemma \ref{lem:K_convex})
Let $F(\Omega, q_1,\ldots,q_k)$ be defined as
\[
F(\Omega, q_1,\ldots,q_k) = \sum_{i=1}^{N(\Omega)} \sum_{\mathbf{m}\in \{0,1\}^{|V_{r(\Omega,i)}|  }} 
	q_{r(\Omega,i)}\left( \mathbf{m} \right) 
		\cdot 
	I\left( X_{\Omega \cap \cT(\Omega,i,\mathbf{m})} ; Y_{\Omega \cap \cR(\Omega,i,\mathbf{m})} | X_{(\Omega \cap \cT(\Omega,i,\mathbf{m}))^c} \right),
\]
and consider the set $\sC(\Omega) = \{(R, P, q_1, \ldots, q_k) : R \le F(\Omega, q_1, \ldots, q_k) \}$.
Choose vectors $(R_1, P_1, q_{1,1},\ldots, q_{k,1})$
and $(R_2, P_2, q_{1,2},\ldots, q_{k,2} )$ in $\sC(\Omega)$. We will show that $\sC(\Omega)$ is a convex set by showing that for any $\gamma \in [0,1]$
the vector $\gamma(R_1, P_1, q_{1,1},\ldots, q_{k,1}) + (1-\gamma)(R_2, P_2, q_{1,2},\ldots, q_{k,2})$ is also in $\sC(\Omega)$. Notice that $R_1 \le F(\Omega, q_{1,1},\ldots, q_{k,1})$ and $R_2 \le F(\Omega, q_{1,2},\ldots, q_{k,2})$ if and only if $(R_1, P_1, q_{1,1},\ldots, q_{k,1}) \in \sC(\Omega)$
and $(R_2, P_2, q_{1,2},\ldots, q_{k,2}) \in \sC(\Omega)$. Therefore
\begin{align*}
\gamma R_1 +  (1-\gamma) R_2 
& \le \gamma F(\Omega,q_{1,1},\ldots, q_{k,1}) + (1-\gamma) F(\Omega, q_{1,2},\ldots, q_{k,2})
\\
& \stackrel{(a)}{\le} F(\Omega, \gamma q_{1,1} + (1-\gamma)q_{1,2},\ldots, \gamma q_{k,1} + (1-\gamma)q_{k,2})
\end{align*}
where $(a)$ is due to the fact that $F(\Omega, q_1,\ldots,q_k)$ is a linear function of $(q_{1},\ldots,q_k)$.  Thus, $\sC(\Omega)$ is a convex set for any $\Omega \subseteq V$. 

In addition, consider the set $\sD$ of $(R, P, q_1,\ldots, q_k)$ satisfying \eq{prob3:c1} - \eq{prob3:c5}. Since \eq{prob3:c1} - \eq{prob3:c5} are linear on $(R, P, q_1,\ldots, q_k)$, $\sD$ is also a convex set, and as a result, $\cK= \cap_{\Omega \subseteq V \backslash\{D\} } \sC(\Omega) \cap \sD$ is a convex set.
\end{proof}

Having proved that \eq{eq:optim_prob3} is a convex program, in order to use \Pref{prop:convopt} we need to check that the conditions of Definitions \ref{def:1}, \ref{def:2}, and  \ref{def:3} are satisfied. Part (i) of Definition \ref{def:1} follows from the linearity of the objective in \eq{eq:optim_prob3}. For part (ii) of Definition \ref{def:1} we specify an infeasibility measure $\Infeas(\cdot):\R^{2+\sum_{i=1}^k 2^{|V_i|}} \to \R$ as follows\footnote{With a slight abuse of notation we use $\max$ on vector quantities by taking the maximum over the components of the vector.}:
\begin{align}
\label{eq:infeas}
\Infeas\big(R&, P, q_{1}(\overbrace{0,\ldots,0}^{|V_1| \, \text{many}} ), \ldots, q_{1}(\overbrace{1,\ldots,1}^{|V_1| \, \text{many}}), 
	\ldots, q_{k}(\overbrace{0,\ldots,0}^{|V_k| \, \text{many}}), \ldots, q_{k}(\overbrace{1,\ldots,1}^{|V_k| \, \text{many}})\big) \nonumber\\
=& \max\Bigg\{0, -q_{1}(0,\ldots,0), \ldots, -q_{1}(1,\ldots,1), \ldots, -q_{k}(0,\ldots,0), \ldots, -q_{k}(1,\ldots,1), \nonumber \\
& \sum_{\bm \in \{0,1\}^{|V_1|}}q_1(\bm)-1, \ldots, \sum_{\bm \in \{0,1\}^{|V_k|}}q_k(\bm)-1, \nonumber\\
& \bigg\{ \sum_{\mathbf{m}\in \{0,1\}^{|V_i|} :\, (m_j)_{j\in V_i \cap V_l}=\mathbf{m}_1} 
			q_{i}\left(\mathbf{m}\right) 
		-
		\sum_{\mathbf{m}\in \{0,1\}^{|V_l|} :\, (m_j)_{j\in V_i \cap V_l}=\mathbf{m}_1} 
			q_{l}\left(\mathbf{m}\right): \bm_1\in \{0,1\}^{|V_i\cap V_l|}, \nonumber\\
& i\ne l, i,l \in \{1,\ldots,k\}\bigg\}, \sum_{i=1}^{|V|} \sum_{\bm \in \{0,1\}^{|V_{s(i)}|}:m_{t(i)}=1} q_{s(i)}(\bm) - P, C_{min} - R, \nonumber\\
& R-\min_{\Omega \in V\backslash \{D\}} F(\Omega\cup \{S\}, q_1,\ldots,q_k)\Bigg\}
\end{align}

The conditions of part (ii) of Definition \ref{def:1} are verified in the following lemma.
\begin{lem}
\label{lem:separation}
For a given vector $(R, P, q_1,\ldots, q_k) \in \R^{2+\sum_{i=1}^k 2^{|V_i|}}$ and any $\ep > 0$ we can either  (a) determine in polynomial time on $|V|2^{\max_i|V_i|}$
if $\Infeas(R, P, q_1,\ldots, q_k)\le \ep$ and if not (b) 
find in polynomial time on $|V|2^{\max_i|V_i|}$ a vector $\bc \in \R^{2+\sum_{i=1}^k 2^{|V_i|}}$, such that for every $(R', P', q_1',\ldots, q_k')$ satisfying $\Infeas((R', P', q_1',\ldots, q_k'))\le \ep$, $\bc^T (R', P', q_1',\ldots, q_k') < \bc^T (R, P, q_1,\ldots, q_k)$.
\end{lem}

\begin{proof}(Lemma \ref{lem:separation})
Part (a) requires checking that each of the arguments of the $\max$ of (\ref{eq:infeas}) is smaller than or equal to $\ep$ in polynomial time. All the terms except for the last one are linear functions of $(R, P, q_1,\ldots, q_k)$ and can be easily computed in time that is polynomial on $|V| 2^{\max_i |V_i|}$. The last term can be compared to $\ep$ by performing a minimization over $\Omega$ of $F(\Omega, q_1,\ldots, q_k)$. As shown in \cite{PE11}, with independent encoding at the nodes (i.e. $X_S, X_1,\ldots, X_n$ independent) and for any channel model where the output variables $Y_i$ are conditionally independent conditioned on $X_S,X_1,\ldots, X_n$, the function $I\left( X_{\Omega \cap \cT(\Omega,i,\mathbf{m})} ; Y_{\Omega \cap \cR(\Omega,i,\mathbf{m})} | X_{(\Omega \cap \cT(\Omega,i,\mathbf{m}))^c} \right)$ is a submodular function of $\Omega$ for any $i=1,\ldots, N(\Omega)$. Since for fixed $q_1,\ldots, q_k$, $F(\Omega, q_1,\ldots, q_k)$ is a linear combination of submodular functions of $\Omega$, it is also a submodular function of $\Omega$. In addition, assuming that $F(\Omega, q_1,\ldots, q_k)$ can be computed in polynomial time on $|V| 2^{\max_i |V_i|}$ (which holds in the Gaussian channel model case), the minimization over $\Omega$ of $F(\Omega, q_1,\ldots, q_k)$ can be computed in polynomial time on $|V| 2^{\max_i |V_i|}$ \cite{PE11}.

We now focus on condition (b). In this case $\Infeas(R, P, q_1,\ldots, q_k) > \ep$, meaning that at least one of the arguments of the $\max$ of (\ref{eq:infeas}) is larger than $\ep$. Let $f_i(R, P, q_1,\ldots, q_k)$ be the $i$th argument of the $\max$ in (\ref{eq:infeas}), and assume that $f_i(R, P, q_1,\ldots, q_k) > \ep$. We can always find a vector $\bc_i$ satisfying $\bc_i^T (R', P', q_1',\ldots, q_k') < \bc_i^T (R, P, q_1,\ldots, q_k)$ for all $(R', P', q_1',\ldots, q_k')$ for which $\Infeas(R', P', q_1',\ldots, q_k') \le \ep$ by setting $\bc_i = \nabla f_i(R, P, q_1,\ldots, q_k)$, where $\nabla$ is the gradient operator. This choice of $\bc_i$ is the normal of a hyperplane that separates $(R, P, q_1,\ldots, q_k)$ from the set $\{(R', P', q_1',\ldots, q_k'):\Infeas(R', P', q_1',\ldots, q_k') \le \ep\}$. In addition, since $f_i(R, P, q_1,\ldots, q_k)$ is a linear function of $(R, P, q_1,\ldots, q_k)$, its gradient is constant (i.e. independent of $(R, P, q_1,\ldots, q_k)$) and hence $\bc_i$ can be precomputed. The complexity of computing $\bc$ is that of determining which argument of the $\max$ is larger than $\ep$ and setting $\bc=\bc_i$. Since the $\max$ has order $|V| 2^{\max_i |V_i|}$ arguments, condition (b) is satisfied.

\end{proof}

\noindent
(Proof of \Tref{mainthm})
Having proved these preliminary results, we are ready to prove the main result of this section.
The optimization of Problem 3 can be solved in polynomial time on the size of the problem.
The proof uses \Pref{prop:convopt}, which requires verifying the convexity of the problem together with the conditions of polynomial computability, polynomial growth, and polynomial boundedness of the feasible set. Convexity was proved in Lemma \ref{lem:K_convex}, while polynomial computability was shown in \Lref{lem:separation}. Polynomial growth follows from the fact that $F(\Omega,q_1, \ldots, q_k)$ is a linear function of $q_1, \ldots, q_k$, while the objective and remaining terms that define the infeasibility measure are linear on $(R, P, q_1, \ldots, q_k)$. Finally, to check that if feasible set is polynomially bounded, we note that the feasible set is a subset of the hypercube
 \begin{align*}
 \{(R, P, & q_1(0,\ldots, 0), \ldots, q_1(1, \ldots, 1), \ldots, q_k(0,\ldots, 0), \ldots, q_k(1, \ldots, 1)):  \\
  & 0 \le (R,q_1(0,\ldots, 0), \ldots, q_1(1, \ldots, 1), \ldots, q_k(0,\ldots, 0), \ldots, q_k(1, \ldots, 1)) \le (R_{\max},  |V| , \bone_1, \ldots, \bone_k) \}
 \end{align*}
where $R_{\max}$ is the maximum rate that is achievable in the network when we use full-duplex communication. It follows that the feasible set is contained in the Euclidean ball centered at the origin with radius $\|(R_{\max},|V|, \bone_1, \ldots, \bone_k)\|_2$, which can be easily checked to grow polynomially on the size of the problem.



\subsection{Heuristic for Finding a Valid Node Grouping}
\label{heuristicsec}

In order to use Problem 3 to find half-duplex schedules, it is necessary to first find subsets of nodes $V_1,...,V_k$ satisfying P1 and P2.
Finding a good choice of the node grouping $V_1,..., V_k$ is not a simple task, and a brute force approach is computationally impractical.
As we will show in Section \ref{appsec}, for specific classes of networks it is in fact easy to find a good node grouping.
For the general case, we describe a heuristic that generates sets $V_1,...,V_k$ satisfying properties P1 and P2, but with no guarantee on the value of $\max_j |V_j|$.

We start by replacing P1 with a list of sufficient conditions.
For a node group $V_\ell$, we let $\{ i \in V_\ell \, : \, (i,j) \in E \text{ for some } j \in V_\ell\}$ be the set of the transmitter nodes in $V_\ell$ and $\{ j \in V_\ell \, : \, (i,j) \in E \text{ for some } i \in V_\ell\}$ be the set of receiver nodes in $V_\ell$.
Then, we notice that a node grouping $V_1,...,V_k$ satisfying
\begin{itemize}
\item every node $i \in V$ is a transmitter in at least one $V_\ell$,
\item if $i \in V_\ell$ is a transmitter node in $V_\ell$, then every node $j \in V$ for which $(i,j) \in E$ is also in $V_\ell$
\item if $j \in V_\ell$ is a receiver node in $V_\ell$, then every node $i \in V$ for which $(i,j) \in E$ is also in $V_\ell$
\end{itemize}
must also satisfy P1.
To see this, consider any cut $\Omega$, and any connected component $V_{\Omega,i}$ in its corresponding cut graph $G_{\Omega}$.
If $V_{\Omega,i} = \{v\}$, then clearly $V_{\Omega,i}$ is a subset of the of the $V_j$ which contains $v$ as a transmitter.
If $|V_{\Omega,i}| \geq 2$, then there must be nodes $u,v \in V_{\Omega,i}$ such that $u \in \Omega$, $v \in \Omega^c$ and $(u,v) \in E_{\Omega}$.
Let $V_j$ be a node group where $u$ is a transmitter.
Any node $w$ such that $(u,w) \in E$ is also in $V_j$.
In particular, any node $w \in V_{\Omega,i}$ such that $(u,w) \in E_{\Omega}$ must be in $V_j$ and it must in fact be a receiver node in $V_j$.
Thus, any node $z \in V_{\Omega,i}$ for which there exists a node $w \in V_{\Omega,i}$ such that $(z,w) \in E_{\Omega}$ and $(v,w) \in E_{\Omega}$ must also be in $V_j$.
Inductively, we see that any path in $V_{\Omega,i}$ starting at $v$ must be in $V_j$, which implies that $V_{\Omega,i} \subset V_j$.

This set of sufficient conditions for P1 suggests a simple rule for the node grouping.
One starts by picking an arbitrary node $v \in V$ and setting it as a transmitter in the first node group $V_1$.
Then, any node $u$ such that $(v,u) \in E$ is also included in $V_1$, and is thus a receiver in $V_1$.
Next, any node $w$ such that $(w,u) \in E$ for a receiver node $u$ already in $V_1$ is also included in $V_1$.
This process can be repeated, until no more nodes need to included in $V_1$.
We can then move to another node $z$ which is not a transmitter in $V_1$, and make it a transmitter in the second node group $V_2$.
Nodes are added to $V_2$ in the same way they were added to $V_1$, and by repeating this process until all nodes are transmitters in at least one node group $V_j$, we generate a node grouping $V_1,...,V_k$ satisfying P1.

This rule for choosing a node grouping satisfying P1, although quite simple, can be sufficient in some cases.
However, if our network graph $G$ is dense, this rule has the tendency to make the node groups very large, and can be quite suboptimal.
In Section \ref{linesec}, we consider an example of a network where this rule would generate a single group $V_1 = V$, but a more careful selection rule yields $V_1,...,V_k$ of much smaller size.

Next, we describe how, given sets $V_1,...,V_k$ satisfying property P1, one can find new sets $C_1,...,C_K$ (where $K \ne k$) satisfying P2.

%
%
%
For sets $V_1, V_2, \ldots, V_k \subseteq V$ satisfying P1, let $\tilde{G}=(V,\tilde{E})$ be an undirected graph satisfying $E \subseteq \tilde{E}$ and $(u,v) \in \tilde{E}$ if $u,v \in V_i$ for some $i$. That is, $\tilde{G}$ is obtained from $G$ by removing edge orientations, and by adding edges to make each $V_i$, $i=1,\ldots,k$, a clique. In addition, let $(C,T)$ be a {\em tree decomposition} of $\tilde{G}$, that is, a family $C=\{C_1,\ldots, C_K\}$ of subsets of $V$ and a tree $T$ whose nodes are the subsets $C_i$ satisfying the following properties \cite{Die05}:
\begin{enumerate}[(1)]
\item The union of all sets $C_i$ equals $V$. That is, each graph vertex is associated with at least one tree node.
\item For every edge $(v, w)$ in the graph $\tilde{G}$, there is a subset $C_i$ that contains both $v$ and $w$. That is, vertices are adjacent in the graph only when the corresponding subtrees have a node in common.
\item If $C_i$ and $C_j$ both contain a vertex $v$, then all nodes $C_k$ of the tree in the (unique) path between $C_i$ and $C_j$ contain $v$ as well. That is, the nodes associated with vertex $v$ form a connected subset of $T$. This is also known as coherence, or the running intersection property. It can be stated equivalently that if $C_i, C_j$ and $C_k$ are nodes, and $C_k$ is on the path from $C_i$ to $C_j$, then $C_i \cap C_j \subseteq C_k$.
\end{enumerate}
We are interested in a tree decomposition of $\tilde{G}$ because it can be used to define a joint distribution that factorizes over the sets $C_i$. 
Let $\mathcal{S}=\{C_i \cap C_j: (i,j) \text{ is an edge in }T\}$ be a multiset, i.e. $\mathcal{S}$ can have repeated elements. 
For a probability distribution $q_l$ on the elements of $C_l$, we let $q_{l:S\subset C_l}$ be the marginal associated with the subset of elements $S \subset C_l$.
Then we can define a joint distribution with alphabet $\mathcal{M}^{|V|}$ for arbitrary $\mathcal{M}$ by
\begin{equation}
q(m_1,\ldots, m_n) = \frac{\prod_{i=1}^K q_i(m_j:j\in C_i)}{\prod_{S\in \mathcal{S}} q_{l:S\subseteq C_l}(m_j: j \in S)},
\label{eq:joint_pmf}
\end{equation}
for $\{q_i(\cdot)\}_{i=1}^K$ satisfying
\begin{align*}
&  0 \le q_i\left(\mathbf{m}\right), \mathbf{m}\in \mathcal{M}^{|C_i|},  
		\ \text{ for } i=1,2 \ldots, K \\
&   \sum_{\mathbf{m}\in \mathcal{M}^{|C_i|}} 
		q_i\left(\mathbf{m}\right) = 1,
			\ \text{ for } i=1,2, \ldots, K \\
&   \sum_{\mathbf{m}\in \cM^{|C_i|} :\, (m_j)_{j\in C_i \cap C_l}=\mathbf{m}_1} 
			q_{i}\left(\mathbf{m}\right) 
		=
		\sum_{\mathbf{m}\in \mathcal{M}^{|C_l|} :\, (m_j)_{j\in C_i \cap C_l}=\mathbf{m}_1} 
			q_{l}\left(\mathbf{m}\right)  \\
& \text{ for all } \mathbf{m}_1 \in \mathcal{M}^{|C_i \cap C_l|}, \text{ for } i\ne l \text{ and } i, l \in \{1,2, \ldots, K\}.
\end{align*}
Therefore, given sets $V_1,...,V_k$ that satisfy property P1, the sets $C_1,...,C_K$ corresponding to the nodes of any tree decomposition of $\tilde G$ satisfy properties P1 (since each $V_i$ is a subset of at least one $C_j$) and P2.
Tree decompositions of the graph $\tilde G$ can be found in many ways, and we refer to \cite{tree} for a more detailed treatment of the subject.
In particular, the method known as node elimination guarantees that the total number of nodes in the tree decomposition, $K$, is at most $|V|$.
Thus, we have shown the following:

\begin{thm}
Consider a wireless relay network defined by a graph $G = (V,E)$ and suppose we have node subsets $V_1,...,V_k$ satisfying P1, and the associated undirected graph $\tilde G$.
Then, if $(C,T)$ is a tree decomposition of $\tilde G$ with nodes $\{C_1,...,C_K\}$ obtained through node elimination, $\ol C_{\iid}$ can be computed with a time complexity that is polynomial in $|V| 2^{\max_i |C_i|}$.
\end{thm}

In particular, we point out that if we can find a tree decomposition for which $\max_i|C_i|$ is at most logarithmic in $|V|$, Problem 3 yields a polynomial-time algorithm for generating an approximately optimal half-duplex schedule.
The width of a tree decomposition is defined as the size of its largest set $C_i$ minus one and the {\em treewidth} $\text{tw}(\tilde G)$ of graph $\tilde G$ is the minimum width among all possible tree decompositions of $\tilde G$.
Therefore, the best time complexity achieved by this approach is polynomial $|V| 2^{\text{tw}(\tilde G)}$.
However, finding the tree decomposition whose width is $\text{tw}(\tilde G)$ is known to be an NP-complete problem \cite{tree}.

%% file: applications.tex
\section{Applications of the Main Results}
\label{appsec}

In this section we first apply the result in Theorem \ref{mainthm} to two special classes of wireless Gaussian networks.
In both cases, Theorem \ref{mainthm} implies that Problem 3 can be computed with a time complexity that is polynomial in the network size.
Then, we describe how the main results can be slightly modified in order to address different network models.

\subsection{Layered Networks} \label{layeredsec}

Consider a layered network consisting of a set of nodes $V$. Set $n = |V|$. 
We assume the network has $L$ layers
and we denote the nodes in layer $\ell$ by $\cV_\ell$ for $\ell = 1,2,\ldots, L$. In the first layer
there is only the source node $\cV_1 = \{S\}$ and in last layer there is only the destination node $\cV_L = \{D\}$. 

\begin{thm}
\label{thm:layered}
For a half-duplex layered wireless relay network, Problem 3 can be used to find a constant-gap optimal transmit/receive schedule in polynomial time, provided that the $|\cV_i|$s grow at most logarithmically with $|V|$. 
\end{thm}

\begin{proof}
Following the discussion in Section \ref{heuristicsec}, we will first find sets $V_1,V_2,..., V_k$ satisfying P1, and then we will take the nodes of a tree decomposition of $\tilde G$, to make sure both P1 and P2 are satisfied.

We let $k = L-1$, and we set $V_i = \cV_i \cup \cV_{i+1}$, for $i=1,...,L-1$.
For an arbitrary cut $\Omega$, it is easy to see that the graph $G_{\Omega}$ has connected components 
$G_{\Omega,i} = (V_{\Omega,i},E_{\Omega,i})$, where $V_{\Omega,i} = (\cV_i \cap \Omega) \cup  (\cV_{i+1} \cap \Omega^c)$ and $E_{\Omega,i} = \{(i,j) \in E_{\Omega} : i,j \in V_{\Omega,i} \}$, for $i=1,...,N(\Omega) = L-1$.
Hence, for any $i=1,...,L-1$, $V_{\Omega,i} \in V_i$.

The tree decomposition of $\tilde G$ will be given by $(C,T)$, where $C = (C_1,...,C_{L-1})$, $C_i = V_i$ for $i=1,...,L-1$, and $T = \{ (V_1,V_2), (V_2,V_3),...,(V_{L-2},V_{L-1}) \}$.
It is clear that properties (1) and (2) of a tree decomposition are satisfied.
To verify (3), we notice that $C_i \cap C_j \ne \emptyset$ for $i < j$ if and only if $j = i+1$.
Therefore it is clear that all nodes in $C_i \cap C_{i+1}$ are part of any $C_k$ in the path from $C_i$ to $C_{i+1}$ (which is given by just $C_i$ and $C_{i+1}$).
Thus, by Theorem \ref{mainthm}, 
Problem 3 can be used to find an approximately optimal transmit/receive schedule in time that is polynomial in $|V| 2^{\max_i |\cV_i|}$.
\end{proof}

\subsection{Line Network with One and Two-Hop Links} \label{linesec}

In this section we consider applying Theorem \ref{mainthm} to a class of non-layered networks, which can be seen as a generalization of line networks.
A line network with one and two-hop links is given by a set of nodes $V = \{1,2,...,n\}$ and edge set $\{(1,2),(2,3),...,(n-1,n)\} \cup \{(1,3),(2,4),...,(n-2,n)\}$, and the source and destination are given by $1$ and $n$ respectively.
If we define $V_i = \{i,i+1,i+2,i+3\}$ for $i=1,2,...,n-3$, we have the following result.

\begin{thm} \label{thm:half-duplex-line}
For a half-duplex line network with one and two-hop links, Problem 2 can be used to find a constant-gap optimal transmit/receive schedule in polynomial time in $|V|$.
\end{thm}

\begin{proof}
Once again we follow the discussion in Section \ref{heuristicsec} to first find sets $V_1,V_2,..., V_k$ satisfying P1 and then take the nodes of a tree decomposition of $\tilde G$ to make sure both P1 and P2 are satisfied.

For an arbitrary cut $\Omega$, consider the graph $G_{\Omega}$.
We claim that any connected component of this graph is contained in $V_i$, for some $i \in  \{1,...,n-3\}$.
To see this, we notice that $G_\Omega$ can be seen as a subgraph of the graph $G'$ shown in Figure \ref{bipartitefig},
	\begin{figure}[ht] \center
		\includegraphics[height=4cm]{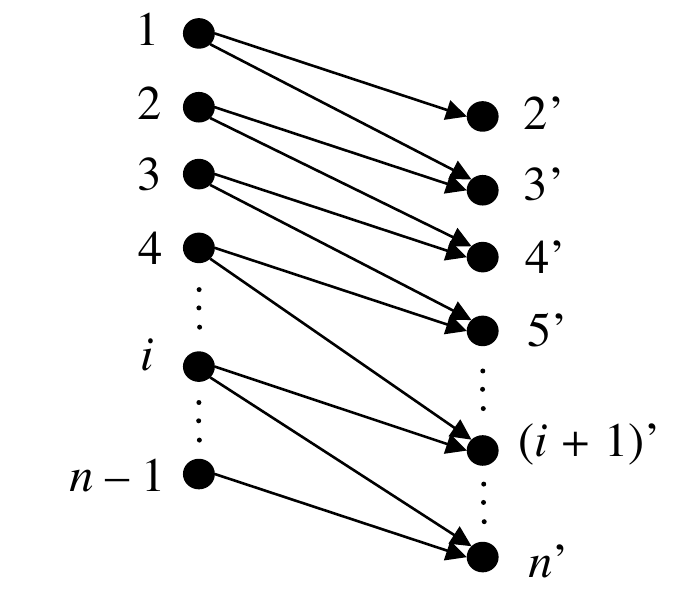}
	\caption{Graph $G'$ from which $G_\Omega$ can be obtained for any $\Omega$}
	\label{bipartitefig}
	\end{figure}
by keeping only the nodes on the left corresponding to the nodes in $\Omega$ and nodes on the right corresponding to $\Omega^c$.
It is then clear that any connected component in $G_\Omega$ must correspond to the nodes on an undirected path in $G'$.
However, it is clear that any undirected path in $G'$ containing more than $4$ nodes must contain nodes $j$ and $j'$, for some $j \in \{1,...,n\}$.
We conclude that any connected component of $G_\Omega$ must be contained in a subset of nodes $\{i,i+1,i+2,i+3\}$ for some $i \in \{1,...,n-3\}$.

The tree decomposition of $\tilde G$ will be given by $(C,T)$, where $C = (C_1,...,C_{n-3})$, $C_i = V_i$ for $i=1,...,n-3$, and $T = \{ (V_1,V_2), (V_2,V_3),...,(V_{n-4},V_{n-3}) \}$.
It is clear that properties (1) and (2) of a tree decomposition are satisfied.
To verify (3), we simply notice that each node $i$ belongs only to the subsets $V_{i-3},V_{i-2},V_{i-1},V_i$ (or just those whose index is greater than zero), and in $T$ we have the path $(V_{i-3},V_{i-2}),(V_{i-2},V_{i-1}), (V_{i-1},V_{i})$.
Thus, the nodes in $C$ associated with $i$ form a connected subset of $T$.
By Theorem \ref{mainthm}, 
Problem 3 can be used to find an approximately optimal transmit/receive schedule in time that is polynomial in $|V| 2^{\max_i |C_i|} = O(|V|)$.
\end{proof}

%% file: extensions.tex
\subsection{Different Network Models}
\label{extsec}

In this section, we discuss how our main results can be applied to other network models.
In essence, \Tref{mainthm} can be extended to any network model where the cut capacities can be computed in polynomial time on $|V|$ and the cut set value $I(X_\Omega;Y_{\Omega^c} | X_{\Omega^c})$ for a product distribution of the transmit signals is a submodular function of $\Omega$. 
In particular, in networks for which the capacity is in fact known to be the cut-set value evaluated with independent transmit signals at each node, our results can be used to efficiently compute the half-duplex capacity\footnote{In channel models where the received signals are affected by noise, our results require that the noise processes are independent across nodes}.
For example, this is the case of the linear deterministic channel model \cite{ADT11}, which we discuss below. 

Consider the general linear deterministic network model from \cite{ADT11}, 
where each node $u$ transmits a vector $\bX_u[t] \in \F^k$ at time $t$, for some $k \in \N$, and each node $v$ receives
\al{
\bY_v[t] = \sum_{u : (u,v) \in E} \mathbf{G}_{u,v} \bX_u[t],
\label{eq:detIO}
}
where all operations are performed over some prime field $\F$.
The ADT networks \cite{ADT11} correspond to the special case where each link $(u,v)$ has an associated non-negative integer channel gain $n_{u,v}$, and the channel matrices are defined as $\mathbf{G}_{u,v} = \mathbf{S}^{k-n_{uv}}$, where $k=\max_{u,v} n_{uv}$ and the shifting matrix $\mathbf{S}$ is a $k \times k$ matrix given by
\[
\mathbf{S}=\left(
\begin{array}{ccccc}
0 & 0 & 0 & \cdots & 0\\
1 & 0 & 0 & \cdots & 0\\
0 & 1 & 0 & \cdots & 0\\
\vdots & \ddots & \ddots & \ddots & \vdots\\
0 & \cdots & 0 & 1 & 0\\
\end{array}
\right).
\]
In the full-duplex scenario, for a given cut $\Omega$ of the network (where $S \in \Omega$ and $D \notin \Omega$), we can stack together the input vectors $\bX_u[t], u\in \Omega$, and output vectors $\bY_v[t], v \in \Omega^c$, and define a transition matrix $\Lambda_\Omega$ that gives the input-output relationship of these vectors according to (\ref{eq:detIO}). 
The cut-set outer bound for these networks, given by
\aln{
\max_{p(\{\bX_v\}_{v\in V})} \min_{\Omega} I(\bX_{\Omega} ; \bY_{\Omega^c} | \bX_{\Omega^c}) = 
\max_{p(\{\bX_v\}_{v\in V})} \min_{\Omega} H(\bY_{\Omega^c} | \bX_{\Omega} )
}
is maximized by choosing the $\bX_v$s to be independent i.i.d.~vectors with entries chosen uniformly at random from $\F$, in which case  $H(\bY_{\Omega^c} | \bX_{\Omega} )=\rank \left(\Lambda_\Omega\right)$.
Since this outer bound can in fact be achieved, by Theorem 4.1 in \cite{ADT11}, the capacity of (full-duplex) linear deterministic networks is given by
\al{  \label{lindetcap}
\min_{\Omega} \rank \left(\Lambda_\Omega\right).
}

This linear deterministic model can be adapted to a half-duplex setting in a very straightforward manner.
As defined in Section \ref{settingsec}, for each time $t=1,2,...,n$, we partition the node set $V$ into a set of transmitter nodes $Tx[t]$ and a set of receiver nodes $Rx[t]$.
At time $t$, each node $u \in Tx[t]$ transmits a signal $\bX_u[t] \in \F^q$, and a node $u \in Rx[t]$, must transmit $\bX_u[t] = \bv 0$.
Then, at time $t = 1,2,...,n$, the signal received by a node $v \in Rx[t]$ is given by \eq{eq:detIO}.
If, instead, node $v$ is in $Tx[t]$, its received signal is just $\bY_v[t] = \bv 0$.

%
%
%
%
%
%

In the half-duplex scenario, we are interested in the mutual information between nodes in transmitter mode and nodes in the receiver mode.
To this end, it is not difficult to see that, in general, for a set of nodes in the transmitter mode $A$ and a set of nodes in the receiver mode $B$, if the transmit signals $\bX_u$ of all nodes $u \in V$ are independently chosen uniformly at random from $\F^q$, we have
\aln{
I(\bX_A ; \bY_{B} | \bX_{A^c}) = \rank \left(\Lambda_{A,B}\right),
}
where $\Lambda_{A,B}$ is the transfer matrix between the vector of transmit signals $\bX_A$ (obtained by stacking the transmit vectors $\bX_u$ of all $u \in A$) and the vector of received signals $\bY_B$ (similarly obtained by stacking the received vectors $\bY_u$ of all $v \in B$).
The capacity of general linear deterministic networks in \eq{lindetcap} can then be used to obtain the capacity of half-duplex general linear deterministic networks, where we employ the terminology on mode configurations from Section \ref{settingsec}.
The proof of the following result follows the ideas from the proof of Theorem 8.3 in \cite{ADT11}, and here we provide a proof sketch.

\begin{thm} \label{hdlindetthm}
The capacity of a half-duplex linear deterministic network is given by
\al{ \label{hdlindetcap}
\max_{q} \min_{\Omega} \sum_{\bv m} q(\bv m) {{ \emph{ \rank }}} \left(\Lambda_{\Omega \cap \cT(\bv m),{\Omega^c \cap \cR(\bv m)}}\right).
}
\end{thm}

\begin{proof}[Proof Sketch]
Consider a half-duplex linear deterministic network with a $k \times k$ transfer matrix $G_{u,v}$ for each link $(u,v)$.
Now, among all probability distributions $q$ on the mode configuration vectors $\bv m \in \{0,1\}^{|V|}$, consider those with rational entries of the form $q = (q_1/M,q_2/M,...,q_{2^{|V|}}/M)$ where $q_i$s and $M$ are integers.
For each such distribution, we can construct a full-duplex network where, instead of picking a transmit signal from $\F^k$, each node picks a transmit signal from $\F^{kM}$, and each $G_{u,v}$ is replaced by a $Mk \times Mk$ block diagonal matrix $\tilde G_{u,v}$ with $M$ $k \times k$ matrices, where the $\ell$th one is
\begin{itemize}
\item $G_{u,v}$ if $u \in \cT(\bv m)$ and $v \in \cR(\bv m)$, where $\bv m$ is the mode configuration vector whose probability under $q$ is $q_j/M$ and $\sum_{i = 1}^{j-1} q_i < \ell \leq \sum_{i = 1}^{j} q_i$,
\item $\bv 0_{k\times k}$ otherwise.
\end{itemize}
For each of these full-duplex networks, if we let $\tilde \Lambda_{A,B}$ be the transfer matrix between the vector of transmit signals $\bX_A$ and the vector of received signals $\bY_B$ and $\Lambda_{A,B}$ be the transfer matrix between the vector of transmit signals $\bX_A$ and the vector of received signals $\bY_B$ in the original half-duplex network (ignoring half-duplex constraints), its capacity can be seen to be
\aln{ 
\min_{\Omega} \rank \left(\tilde \Lambda_{\Omega,\Omega^c}\right) =
\min_{\Omega}  \sum_{j=1}^{2^{|V|}} q_j \, \rank \left(\Lambda_{\Omega \cap \cT(\bv m),{\Omega^c \cap \cR(\bv m)}}\right) = 
\min_{\Omega} M \sum_{\bv m} q(\bv m) \rank \left(\Lambda_{\Omega \cap \cT(\bv m),{\Omega^c \cap \cR(\bv m)}}\right).
}
%

Next, notice that for any coding scheme for the half-duplex network, it is possible to define a probability distribution $q(\bv m) = n^{-1}|\{ t : Tx[t] = \cT(\bv m) \}|$ on the mode configuration vectors.
Thus, if $M$ is large enough, it is not difficult to see that a given coding scheme can be replicated on one of the full-duplex networks we considered (using, of course, a blocklength $M$ times smaller).
Conversely, a coding scheme for one of these full-duplex networks can be converted into a coding scheme for the original half-duplex network (where we may have to delay signals to preserve causality).
This allows us to conclude that the capacity of the original half-duplex linear deterministic network is given by \eq{hdlindetcap}.
\end{proof}
%
%
%
%
%
%
%
%
%
%
%

Based on Theorem \ref{hdlindetthm},  
we can write a linear program similar to Problem 2 that characterizes capacity of linear deterministic networks.
However, solving it is difficult due to the exponential number of different mode configurations.
To find an equivalent optimization problem of lower time complexity, as we did in Section \ref{mainsec}, we consider finding subsets $V_1,...,V_k$ satisfying properties P1 and P2.
Then we consider rewriting Problem 3 as follows.


%
%

\vspace{3mm}

{\noindent \bf Problem 3':}
\begin{align}
\label{eq:optim_prob3ld}
\underset{R, \{q_i\}_{i=1}^k}{\text{maximize}}& \ R    \\
\text{subject to}  &   \nonumber \\
\label{prob3ld:c0}
&   R < \sum_{i=1}^{N(\Omega)} \sum_{\mathbf{m}\in \{0,1\}^{|V_{r(\Omega,i)}|  }} 
	q_{r(\Omega,i)}\left( \mathbf{m} \right) 
		\cdot
	\rank \left( \Lambda_{\Omega \cap \cT(\Omega,i,\mathbf{m}), \Omega^c \cap \cR(\Omega,i,\mathbf{m})} \right), 
	\nonumber \\
	& \quad \quad \quad  \forall \, \Omega \in 2^V: S \in \Omega, D \not\in \Omega, \\
\label{prob3ld:c1}
	&  0 \le q_i\left(\mathbf{m}\right), \mathbf{m}\in \{0,1\}^{|V_i|},  
		\ \text{ for } i=1,2 \ldots, k \\
\label{prob3ld:c2}
	&  \sum_{\mathbf{m}\in \{0,1\}^{|V_i|}} 
		q_i\left(\mathbf{m}\right) = 1,
			\ \text{ for } i=1,2, \ldots, k \\
\label{prob3ld:c3}
	&  \sum_{\mathbf{m}\in \{0,1\}^{|V_i|} :\, (m_j)_{j\in V_i \cap V_l}=\mathbf{m}_1} 
			q_{i}\left(\mathbf{m}\right) 
		=
		\sum_{\mathbf{m}\in \{0,1\}^{|V_l|} :\, (m_j)_{j\in V_i \cap V_l}=\mathbf{m}_1} 
			q_{l}\left(\mathbf{m}\right)  \\
& \quad \quad \text{ for all } \mathbf{m}_1 \in \{0, 1\}^{|V_i \cap V_l|}, \text{ for } i\ne l \text{ and } i, l \in \{1,2, \ldots, k\}, \nonumber
\end{align}

Showing that Problem 3' indeed is equivalent to finding (\ref{hdlindetcap}) is done by following  similar steps to those in Section \ref{equivproof}.
Moreover, the analysis of the time complexity of solving Problem 3' can follow the steps in Section \ref{computproof}, by noticing that the new function 
\aln{
F(\Omega, q_1,\ldots,q_k) = \sum_{i=1}^{N(\Omega)} \sum_{\mathbf{m}\in \{0,1\}^{|V_{r(\Omega,i)}|  }} 
	q_{r(\Omega,i)}\left( \mathbf{m} \right) 
		\cdot 
		\rank \left( \Lambda_{\Omega \cap \cT(\Omega,i,\mathbf{m}), \Omega^c \cap \cR(\Omega,i,\mathbf{m})} \right),
}
is also a submodular function of $\Omega$ and can be computed in polynomial time on $|V| 2^{\max_i |V_i|}$ for any $\Omega$ and $\{q_i\}_{i=1}^k$.
This yields the following result.

\begin{thm}
\label{thmld}
The capacity of a half-duplex linear deterministic network with node subsets $V_1,...,V_k$ satisfying P1 and P2 can be found by solving Problem 3'.
Moreover, Problem 3' can be solved in a time complexity that is polynomial in $|V| 2^{\max_i |V_i|}$.
\end{thm}


Another class of networks to which the results in this paper can be extended is the class of determinisitc relay networks with state information from \cite{DeterministicWithStates}.
Each node in these networks receives a function of the transmit signals at all the other nodes and of a random network state, and this state is supposed to be known at the destination.
In the case where, given the network state, the received signals are linear functions of the transmit signals, it is shown in \cite{DeterministicWithStates} that the capacity is given by 
\al{ 
\min_{\Omega} E \left[ \rank \left( \Lambda_{\Omega} \right) \right],
}
generalizing the result in (\ref{lindetcap}).
It is not difficult to see that, in this case, analogous results to those in Theorems \ref{hdlindetthm} and \ref{thmld} can be obtained.

%% file: simulations.tex
\section{Numerical Results}
\label{sec:simulations}
In this section we take a step further and numerically evaluate the performance of our algorithms. We show that with our proposed approach we can evaluate the half-duplex cut-set bound and determine the associated half-duplex schedules for networks of larger size than alternative algorithms. In addition, we show that by optimizing the half-duplex schedules we can approximate the performance of full-duplex networks.
We focus on a Gaussian layered relay network, where the channel gains are randomly chosen with i.i.d. $\cC\cN(0,{P})$ distribution. This allows us to use \Tref{thm:layered} to solve Problem 3 for various network sizes as we vary the number of layers $L$ keeping the number of nodes in each layer $|\cV_l|$ constant for $l=2,\ldots, L-1$. While in previous sections we assumed a real Gaussian channel model, all the results can be easily extended to the complex case which we consider in this section, which is of practical relevance for baseband models of wireless systems.

\begin{figure}[ht] \center
\includegraphics[width=12.6cm]{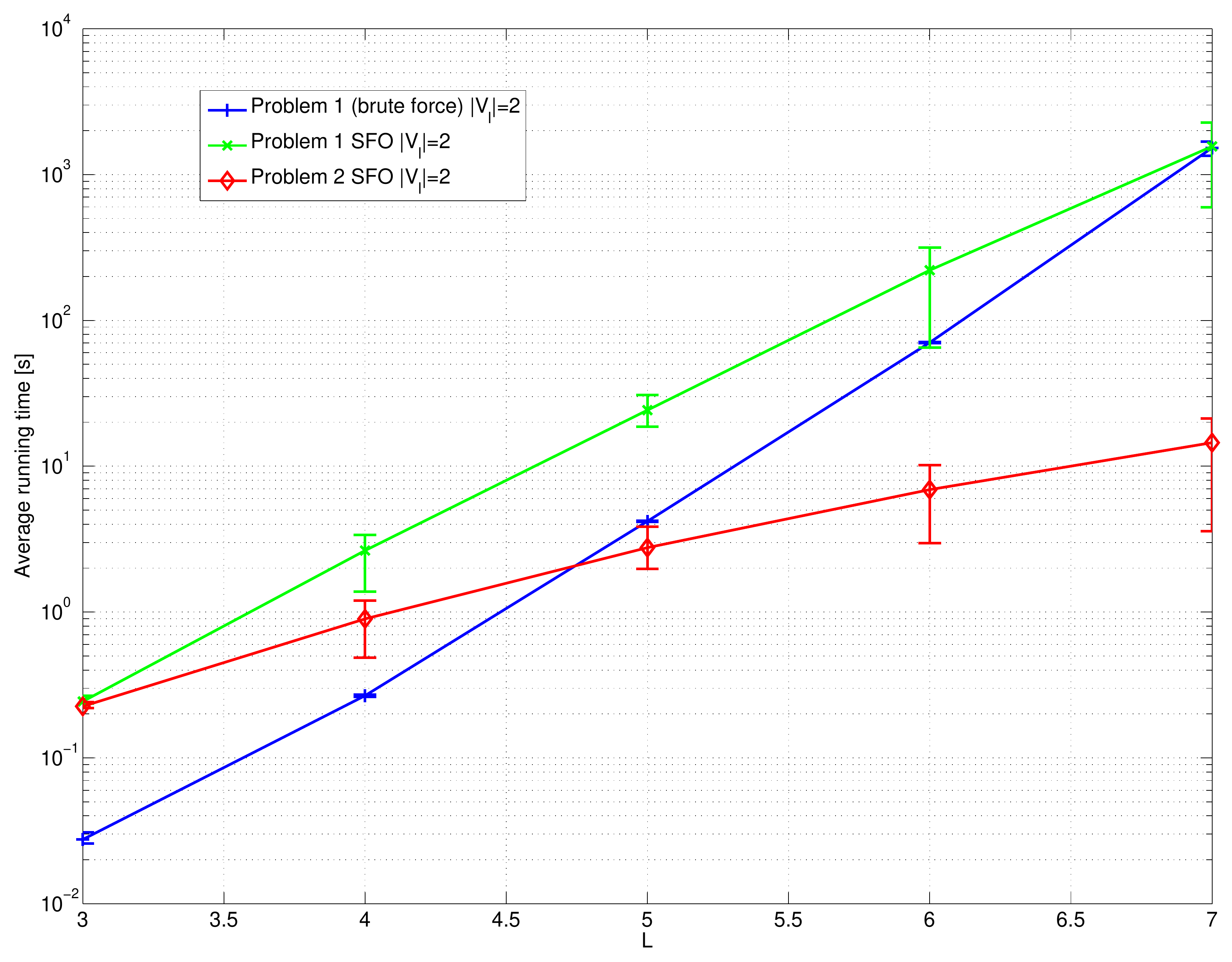}
\caption{Average running times of the algorithms that solve Problem 2 by brute force and by using submodular optimization (SFO) compared to that of the algorithm that solves Problem 3 using \Tref{thm:layered} as a function of the number of layers $L$ with $|\cV_l|=2$ nodes per layer. Also shown are error bars indicating the minimum and maximum running times over 10 random realizations of the network.
}
\label{fig:slopes}
\end{figure}

In order to set a baseline for comparison, for networks of small size, Problem 2 is solved using Matlab's linear optimization function \texttt{linprog}. On the other hand, we solve Problem 3 efficiently using \Tref{thm:layered} and exploiting submodularity with the method proposed by Fujishige \cite{Fuj05,FHI06} and the Matlab package SFO \cite{sfo}. In addition, in order to evaluate the improvements obtained from the reduction in the number of optimization variables in going from Problem 2 to Problem 3, we also solve Problem 2 exploiting submodularity with the SFO package. 


Fig. \ref{fig:slopes} shows the average running times of the various algorithms as the number of layers $L$ increases for fixed number of relays $|\cV_l|=2$ in each layer. The average running time is computed over 10 random realizations of the channel gains drawn i.i.d.\,with distribution $\cC\cN(0,1)$.
We first note that the brute force approach of solving Problem 2 with \texttt{linprog} leads to smaller running times than the approach that uses
SFO. This may be explained by noting that Problem 2 has exponential complexity on the number of nodes, whether the optimization is solved by brute force or by submodular optimization methods. While SFO may simplify the computation of the cut-set constraints, it introduces overheads which grow with the number of variables involved. Fig. \ref{fig:slopes} shows that these overheads dominate the running times for the values of $L$ in the plot, but due to the smaller slope of the curve, the solution with submodular optimization should be faster for large enough $L$. The overheads of SFO also explain why the algorithm that solves Problem 2 by brute force via \texttt{linprog} exhibits smaller running times than the algorithm that solves Problem 3 for networks with 7 or fewer nodes. However, as the number of nodes in the network grow the running times of the algorithm implementing Problem 3 become significantly smaller ($\approx 1/100$ for $L=7$) than those of the other algorithms.  

\begin{figure}[ht] \center
\includegraphics[width=14cm]{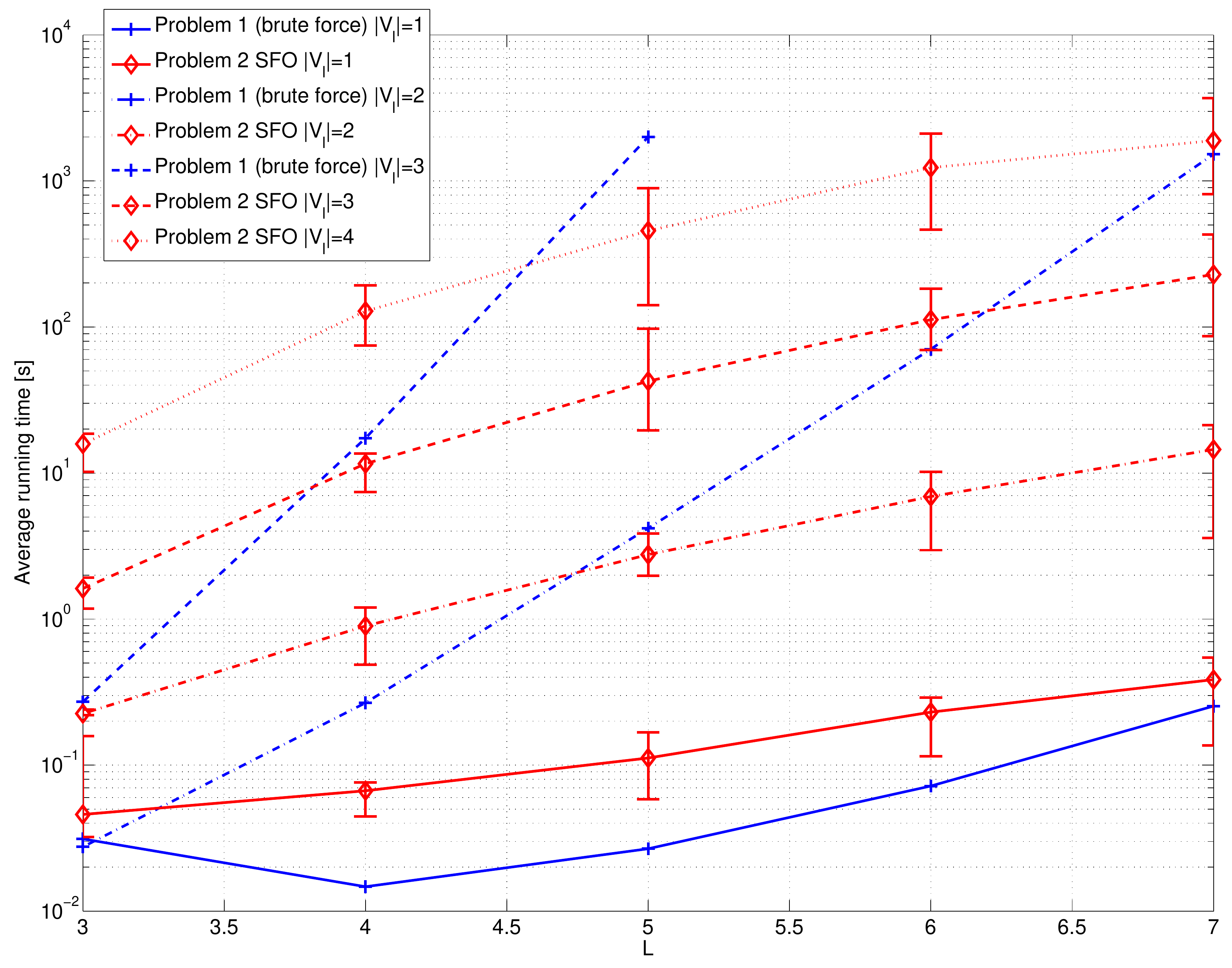}
\caption{Average running times the algorithm that solves Problem 2 by brute force compared to that of the algorithm that solves Problem 3 using \Tref{thm:layered} as a function of the number of layers $L$ for different number of nodes per layer $|\cV_l|$. The average running times are computed over 10 runs with random channel gains. Also shown are error bars indicating the minimum and maximum running times for the 10 runs of the algorithm that solves Problem 3. The algorithm that solves Problem 3 has smaller running times for networks of moderate and large size.}
\label{fig:timing1}
\end{figure}

Figure \ref{fig:timing1} shows the average running times of the algorithms that solve Problem 2 by brute force and Problem 3, as the number of layers $L$ increases for fixed number of relays $|\cV_l|$ in each layer. As before, the average running times are computed over 10 random realizations of the channel gains drawn i.i.d. with distribution $\cC\cN(0,1)$. Since Problem 2 has exponential complexity on the number of nodes, the running times of the algorithm that solves Problem 2 are only given for networks of small size. As seen in the figure, the algorithm implementing Problem 3 not only has smaller running times for networks of moderate size, but also it allows to solve problems for networks of large size for which the other algorithms become impractical due to memory and time requirements.

Figure \ref{fig:timing1} also includes the minimum and maximum running times of the algorithm that solves Problem 3 for 10 random realization of the channel gains. The ratio of the maximum to the minimum running time for each network size is approximately constant for various network sizes and is smaller than 10 in all cases shown in the figure.

\begin{figure}[ht] \center
\includegraphics[width=14cm]{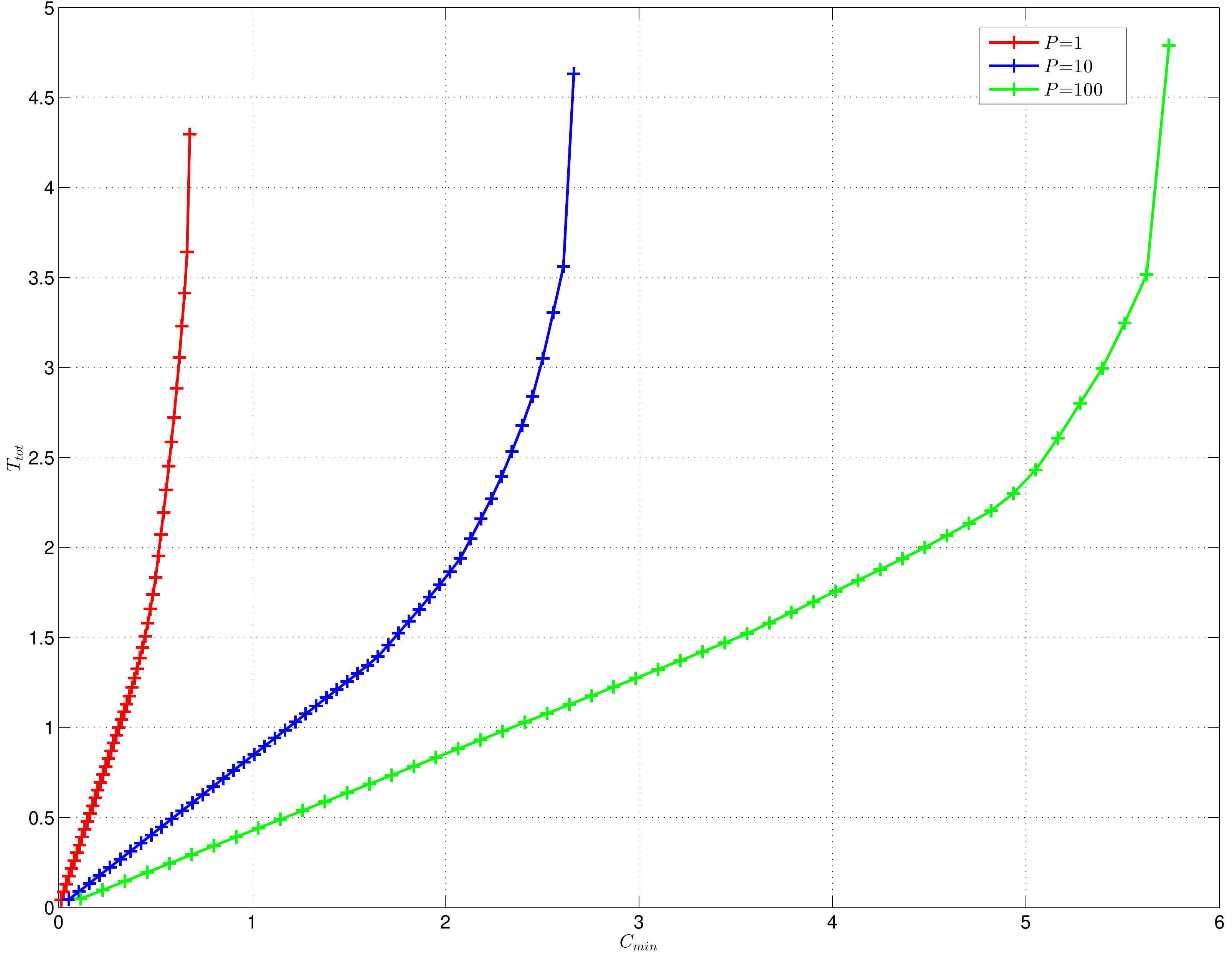}
\caption{Minimum total duty cycle $T_{tot}$ as a function of the minimum required cut-set capacity $C_{min}$ for random networks with $L=4$ layers and $|\cV_l|=3$ nodes per layer. The random channel gains are drawn i.i.d. with $\cC\N(0,{P})$ distribution, with ${P}=1,10,100$. The figure shows that for $C_{min}$ small compared to the half-duplex cut-set capacity the total duty cycle (and hence the power consumption) increases linearly with $C_{min}$. As $C_{min}$ approaches the half-duplex cut-set capacity the total duty cycle increases abruptly.}
\label{fig:PvsC0}
\end{figure}

As discussed in Section \ref{settingsec}, we can choose $\mu_1$, $\mu_2$ and $C_{min}$ in Problems 2 and 3 to find schedules optimizing different objectives. For example, by choosing $\mu_1=0$ and $\mu_2=1$ and we can find the minimum required sum of duty cycles $T_{tot}$ (which is proportional to the required energy consumption of the network) to achieve a given minimum cut-set capacity $C_{min}$. Figure \ref{fig:PvsC0} shows this energy consumption as a function of $C_{min}$ for networks with 4 layers and 3 relays per layer. We study different SNR regimes by changing the variance ${P}$ of the distribution used to generate the channels gains. 

We observe in Figure \ref{fig:PvsC0} that for $C_{min}$ sufficiently smaller than the half-duplex cut-capacity of the network, $T_{tot}$ increases linearly with $C_{min}$. For sufficiently large $C_{min}$ the required total duty cycle $T_{tot}$ increases faster as $C_{min}$ increases. This can be explained by noting that for small $C_{min}$ the optimal schedules only use the relays that have large channel gains, and their corresponding duty cycle is proportionally increased as $C_{min}$ increases. As $C_{min}$ becomes a significant fraction of the half-duplex cut-capacity, more and more nodes need to be activated by making them transmit a positive fraction of time. These additional nodes provide only marginal rate improvements since they do not increase the degrees of freedom of the network and their associated channel gains are likely small compared to those of the nodes activated first. This results in larger increase in total duty cycle (and power consumption) for the same increase of $C_{min}$ as $C_{min}$ approaches the half-duplex cut-capacity.

We next evaluate the benefits of using optimal half-duplex schedules (from solving Problem 3) compared to simpler scheduling approaches for networks with $L=4$ layers, $|\cV_l|=4$ nodes per layer, and i.i.d. $\cC\cN(0,{P})$ channel gains. 
We compare the ratio between the half-duplex and full-duplex cut-set bound (averaged over 10 different random network realizations) for three half-duplex scheduling techniques:
\begin{itemize}
\item Optimized: the half-duplex schedule that results from solving Problem 2
\item Na\"{i}ve: time is divided in timeslots of equal length. 
The nodes in even layers transmit in even timeslots and receive in odd timeslots, while the nodes in odd layers transmit in odd timeslots and receive in even timeslots.
\item Simple random: the network is divided into 2 subnetworks, where half of the nodes in each relay layer are randomly assigned to subnetwork 1, while the other half is assigned to subnetwork 2. The source and destination nodes are part of both subnetworks. Each subnetwork is operated using the na\"{i}ve half-duplex schedule with phase offset of 1 timeslot. That is, half of the nodes in each relay layer transmit at a given time while the other half receive at that time. The source node is always in transmit mode, while the destination node is always in receive mode.

\end{itemize}
\begin{figure}[ht] \center
\includegraphics[width=14cm]{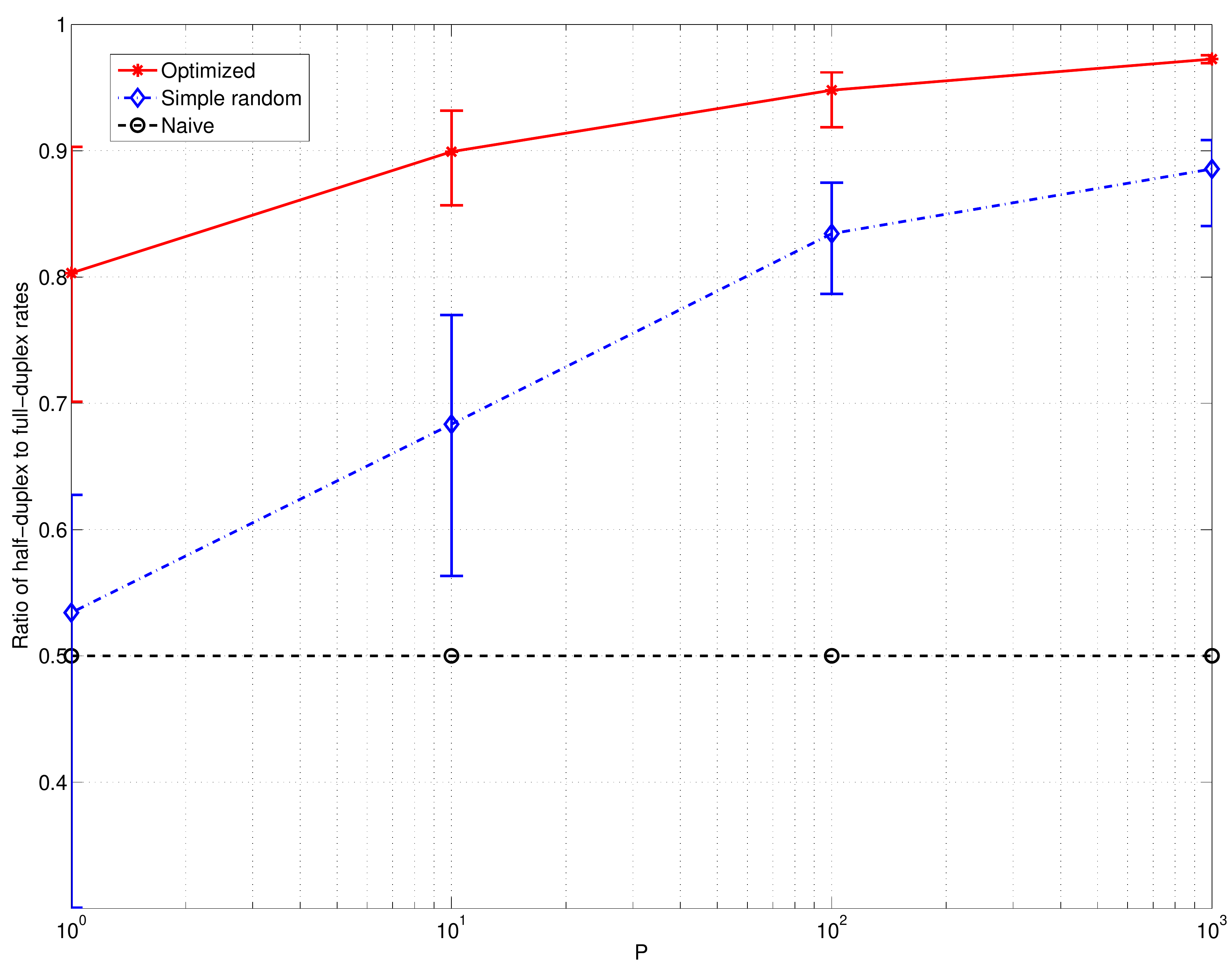}
\caption{Ratio of half-duplex to full-duplex rate for various scheduling approaches as a function of the signal-to-noise ratio $P$. Also shown are error bars indicating the minimum and maximum half-duplex to full-duplex ratios for 10 random realizations of the network.
}
\label{fig:HdFd}
\end{figure}

Fig. \ref{fig:HdFd} shows that the optimized half-duplex schedule allows to obtain a large fraction of the full-duplex performance, which increases with the SNR (represented by ${P}$). 
The simple random schedule outperforms the na\"{i}ve schedule in most cases, but does not match the performance of the optimized schedule for small and moderate SNR. At large SNR, the DoF limits the rates, and the simple random schedule (which is DoF optimal) approaches the performance of the optimized schedule. Note that the curve corresponding to the na\"{i}ve half-duplex schedule is a constant 1/2 due to the fact that the cut-set bound with independent encoding at the nodes for the na\"{i}ve half-duplex schedule is exactly half of the full-duplex cut-set bound. 
We also point out that, as the SNR tends to infinity, the ratio between half-duplex and full-duplex capacity approaches $1$ due to the fact that the DoF in both the half-duplex and full-duplex cases is $1$ whenever the number of nodes per layer is at least $2$.

%% file: conclusions.tex

\newpage

\section{Conclusions}
\label{sec:conc}
We showed that for some classes of Gaussian relay networks an approximation to the half-duplex cut-set bound can be efficiently computed. This computation not only gives a capacity characterization within a constant gap, but also provides ``good'' half-duplex schedules which can be used in communication schemes. By varying some parameters in the optimization program it is possible to obtain schedules that tradeoff rate for lower network power consumption. It was shown by simulations that it may be beneficial to operate the network at rates significantly below the cut-set bound in order to achieve better power efficiency. 

Whenever the complexity of implementing half-duplex schedules is a concern, as is the case in fading channels, it may be of interest to implement simple half-duplex schedules involving few timeslots. The performance of these simple schedules can be evaluated and further optimized using the cut-set bound as a metric, and compared to what is achievable under optimal scheduling using the techniques that we presented. 

We proposed a heuristic approach based on tree decompositions to find the node groupings used in our main theorem. We gave two examples of network classes for which the required tree decompositions are easy to compute. A possible future research direction is to find other network classes leading to efficient tree decompositions (i.e. with small cardinality nodes) where our results can be applied. In addition, it would be of interest to design algorithms for computing these efficient tree decompositions.

%% file: bib.tex

%
%
%
%
%

\bibliographystyle{unsrt}

